\documentclass{article}

\usepackage{amsmath, amsfonts, amssymb, amsthm, txfonts}
\usepackage{color}

\usepackage[margin=1.25in]{geometry}
\usepackage{graphicx}
\usepackage[backref=page]{hyperref}
\usepackage{xfrac}
\graphicspath{ {./figures/} }
\usepackage{subcaption}
\usepackage{lineno}
\usepackage{comment}
\usepackage[utf8]{inputenc}
\usepackage{authblk}
\usepackage{setspace}

\newtheorem{definition}{Definition} [section]
\newtheorem{theorem}{Theorem} [section]
\newtheorem{lemma}{Lemma} [section]
\newtheorem{proposition}{Proposition} [section]
\newtheorem{example}{Example} [section]
\newtheorem{remark}{Remark} [section]

\makeatletter
\newcommand{\leqnomode}{\tagsleft@true\let\veqno\@@leqno}
\makeatother

\title{Dissipative dynamics for infinite lattice systems}
\author[1*$\dag$]{Shreya Mehta}
\author[2]{Boguslaw Zegarlinski}

\affil[1]{Department of Mathematics, Imperial College, London, United Kingdom}
\affil[2]{Université de Toulouse ; CNRS, UPS, F-31062 Toulouse Cedex 9, France}
\date{}
\begin{document}

\maketitle
\begin{abstract}
We study  dissipative dynamics constructed by means of non-commutative Dirichlet forms for various lattice systems with multiparticle interactions associated to CCR algebras. We give a number of explicit examples of such models.  Using an idea of quasi-invariance of a state, we show how one can construct unitary representations of various groups. Moreover in models with locally conserved quantities associated to an infinite lattice we show that there is no spectral gap and the corresponding dissipative dynamics decay to equilibrium polynomially in time.
\end{abstract}
\tableofcontents

\section{Introduction}

The theory of quantum dissipative systems has grown rapidly in the last few decades (see e.g. \cite{CZ} and references therein). Such systems can be described by an equation of the following form
\begin{equation}{\label{1.1}}
    \frac{\partial}{\partial t}P_tf=\mathfrak{L}P_tf,\:P_0=id
\end{equation}
where $\mathfrak{L}$ is a Markov generator and $P_t \equiv e^{t\mathfrak{L}}$ is the associated Markov semigroup. The semigroup $P_t$ is positivity and unit preserving. An easier way to prove closability and dissipativeness of $\mathfrak{L}$ is to consider a symmetric semibounded operator in a Hilbert space. A delicate problem in quantum case is how to simultaneously have symmetry in a Hilbert space associated to a state and positivity preservation of the generated semigroup. One of the solutions to this problem come via proving the closability and Markov property of the pre-Dirichlet form. The theory of Dirichlet forms 
originated in \cite{AHK} for the trace state and fully developed later in \cite{Cip1}; for further development see e.g. references in \cite{Cip3}, \cite{P1},\cite{Z},\cite{CZ}.
In this work, we assume that the generator $\mathfrak{L}$ is constructed using Quantum Dirichlet form given by
\begin{equation}\label{DiForm}
\mathcal{E}_\Lambda(f)\equiv \sum_{j\in \Lambda} \int_{\mathbb{R}} \left(\nu_j\langle \delta_{\alpha_t(X_j) )} (f), \delta_{\alpha_t(X_j)} (f)\rangle_{\omega}+\mu_j\langle \delta_{\alpha_t(X_j^\ast) )} (f), \delta_{\alpha_t(X_j^\ast)} (f)\rangle_{\omega}\right) \eta(t)dt
\end{equation}
for $\nu_j,\mu_j$ and $\Lambda\subseteq\mathbb{Z}^d$, with a scalar product
\[\langle f, g \rangle_{\omega}\equiv lim_{\Lambda\rightarrow \mathbb{Z}^d}\langle f, g \rangle_{\omega_\Lambda} = lim_{\Lambda\rightarrow \mathbb{Z}^d} {Tr} (\rho_\Lambda^{1/2} f^{\ast}
\rho_\Lambda^{1/2} g)\]
where $\omega_\Lambda$ is a finite volume state corresponding to a density matrix 
 $\rho_\Lambda \equiv e^{-\beta U_\Lambda}/ {Tr_\Lambda(e^{-\beta U_\Lambda})} $, and the modular dynamics corresponding to $\omega$ is defined by
\begin{equation}{\label{MO}}\alpha_t (B) = lim_{\Lambda\rightarrow \mathbb{Z}^d} e^{- {it\beta U_\Lambda}} {Be}^{{it\beta U_\Lambda}} .\end{equation}   
The function $\eta(t)$ is an admissible function  in the sense of Park, \cite{P1}, (see Definition (\ref{admissible}) below), and the operators $X_j$'s are chosen appropriately later for specific cases.

By the Beurling Deny theory(\cite{BD}), (with its noncommutative generalisation  \cite{Cip1,Cip3}), providing a one to one correspondence between Markov semigroups and Dirichlet forms, the problem of analysing the generator of the Markov semigroups is translated to studying their corresponding Dirichlet forms.  
%The generator $\mathfrak{L}$ obtained from the Dirichlet form and the corresponding semigroup $P_t$ are called KMS symmetric. 
We say that the semigroup  $P_t$ converges to equilibrium if $lim_{t\rightarrow\infty}P_t(f)=\omega(f)$ where the limit is with respect to the noncommutative $L_p$ spaces or the norm of the algebra. Frequently one can study decay to equilibrium in $L_2$ and $L_p$ spaces using coercive inequalities of Poincare and Log Sobolev type that were studied e.g. in  \cite{OZ,CB} and references therein. 

We discuss construction of Dirichlet forms for models of large interacting systems.   As a starting point we use the quantum harmonic oscillator (refer \cite{CFL}). Given a Hilbert space, say $\mathfrak{h}$ with $\{e_n\}$ as the orthonormal basis, one defines the creation and annihilation operators by 
\begin{equation}{\label{Creation}}A^{\ast}e_n=\sqrt{n+1}e_{n+1}\end{equation}
\begin{equation}{\label{Ann}}Ae_n=\sqrt{n}e_{n-1}\end{equation}
and the associated particle number operator 
\[ Ne_n\equiv A^\ast A e_n=ne_n\]
with dense domains $D(A)=D(A^\ast)=D(\sqrt{N})$.
%where 
%\[D(N)=\left\{\alpha\in \mathfrak{h}:\sum_{n\geq0}|n\alpha_n|^2<\infty\right\}.\]
For $U=N$ we define a density matrix  $\rho=\frac1Z e^{-\beta U}$, with a normalization constant $Z\in(0,\infty)$. In commutative analysis, one considers an Ornstein-Uhlenbeck (O-U) semigroups defined by the Dirichlet generator associated to a Gaussian measure.
The noncommutative generalisation of O-U semigroup is introduced in \cite{CFL},\cite{KP} with the generator (\cite{CM})  given by the extension of the Alicki's theorem
\begin{equation}{\label{Gen}}
 \mathfrak{L}f= \sum_{j}\left(-e^{-\frac{\beta}{2} }[V_j,f] V_j^{\ast}+e^{\frac{\beta}{2} } V_j^{\ast}[V_j,f]\right)
\end{equation}
where $V_j$ are the eigenvectors of the modular operator \eqref{MO}.  
The quantum Markov semigroups generated by the generators of the form \eqref{Gen} were studied in \cite{CM} in the finite dimensional setting %of $C^\ast$ algebras 
where the authors obtained some entropic dissipation inequalities for the necessarily infinite-dimensional Bose O-U semigroup.  

 In case of infinite systems with interaction, in order to be able to define Dirichlet forms and generators on dense domains, the finite speed of propagation of information for the hamiltonian dynamics needs to be satisfied. The case with bounded multi-particle interactions for quantum spin systems was established in \cite{LR}. We generalise it to 
the systems with multiparticle interactions involving unbounded operators.

Some dissipative dynamics for the quantum spin systems on the lattice were earlier discussed in \cite{MZ}, \cite{GM}, \cite{Z}, where they show the existence of dynamics that have an exponential decay to equilibrium in the high-temperature region. 
General systems with quadratic interaction were studied in \cite{OZa}, \cite{BKP} and \cite{FQ}. 
Because of the presence of unbounded operators in the Markov generators, these issues are more complicated than otherwise.

In this work, we construct the models of interacting dissipative systems on a finite dimensional lattice $\mathbb{Z}^d$ with finite range interaction. %on family of finite polynomials in the creation and annihilation operators. 
We consider the multiparticle interaction and generalise the setup of quantum spin systems used e.g. in   \cite{MZ, LR, BR}. For some of our models with locally conserved quantities, we provide the detailed analysis proving that Poincar{\'e} inequality cannot be satisfied and the system converges to equilibrium with polynomial rate of convergence. In particular this extends commutative case considered in \cite{INZ}.

A quantum Brownian motion model provided in \cite{CFL} has no spectral gap and equilibrium state. Some of our models are more general with no spectral gap at the bottom of spectrum of Markov generator along with the existence of an equilibrium state.

The organisation of this paper is as follows. In Section 2, we describe the infinite dimensional quantum system considered in this work and give some basic notations and definitions. We also address the corresponding domain issues of derivations, their adjoints and Dirichlet forms. Section 3 we explore the idea of quasi-invariance of the states to provide the unitary group representations related to certain transformations of the commutation relations. This includes in particular new representations of Minkowski group. %This will be elaborated more later in \cite{MsZ1} and \cite{MsZ2}.

In Section 4, the modular dynamics is constructed and the finite speed of propagation of information for unbounded potentials is proven. In Section 5, the convergence of states and norms in the corresponding $\mathbb{L}_p$ space is established. We continues in Section 6 with definitions of non-commutative Dirichlet forms and corresponding Markov generators. In Section 7, we provide number of explicit models and their Dirichlet forms and Markov generators. Section 8 and Section 9 establishes that there is no spectral gap for some examples, including the case which is invariant with respect to modular operator derivations.

In Section 10 generalises the result of \cite{INZ} by showing the algebraic decay to equilibrium for the corresponding quantum Markov semigroups. 
Finally, in the Appendix, we provide computations for the corresponding $\Gamma_1$ form and discuss the operator carr{\'e} du champ type bounds and some of their decay to equilibrium implications. 

\section{The Infinite Quantum System}
\textbf{The Lattice}\\
For $d\in\mathbb{N}$, let $\mathbb{Z}^d$ be the $d$-dimensional square lattice with the $l_1$ lattice metric $dist(\cdot,\cdot)$ defined by
\[dist(\mathbf{i},\mathbf{j}):=|i-j|_1\equiv\sum_{l=1}^d|i_l-j_l|\]
for $\mathbf{i}=(i_1,\dots,i_d),\mathbf{j}=(j_1,\dots,j_d)\in\mathbb{Z}^d.$ Given $R\in(0,\infty)$, for $\mathbf{i},\mathbf{j}\in\mathbb{Z}^d$ we write $\mathbf{i}\sim \mathbf{j}$ whenever $0\leq dist(\mathbf{i},\mathbf{j})\leq R$ and say that $\mathbf{i}$ and $\mathbf{j}$ are neighbours in the lattice. If $\mathcal{O}$ is a finite subset of $\mathbb{Z}^d$ we will write $\mathcal{O}\subset\subset\mathbb{Z}^d $.\\

%Let $\Lambda$ denote the family of all finite subsets of $\mathbb{Z}^d$ and $\Lambda_n$ be an increasing sequence of finite volume that invades the entire lattice $\mathbb{Z}^d$.

\noindent\textbf{The configuration space}\\
For $\Lambda\subseteq\mathbb{Z}^d$, we associate a separable Hilbert space $\mathcal{H}_\Lambda$ with a property that for bounded sets $\Lambda_1,\Lambda_2\subset\subset\mathbb{Z}^d$, $\Lambda_1\subset\Lambda_2$, we have $\mathcal{H}_{\Lambda_1}\subset \mathcal{H}_{\Lambda_2}$ and $\mathcal{H}_{\mathbb{Z}^d}=\overline{{\cup}_{\Lambda\subset\subset\mathbb{Z}^d}\mathcal{H}_\Lambda}$. 
For a finite set $\mathcal{O}\subset\subset \Lambda$, let $\mathcal{A}_\mathcal{O}\subset\mathcal{L}(\mathcal{H}_\mathcal{O})$ denote the algebra of bounded operators. Let $\mathcal{A}=\cup_{\mathcal{O}\subset\subset\Lambda}\mathcal{A}_\mathcal{O}$ be an inductive limit algebra  \cite{BR}  generated by all local algebras.\\

We assume, we have a family of mutually commuting copies
of creation and annihilation operators satisfying the CCR given by
\begin{equation}{\label{CCR}}
    [A_j,A_j^\ast]=id,\:\:[A_j,A_k^\sharp]=0 , \, j\neq k. 
\end{equation}
where $A_k^\sharp\in\{A_k,A_k^\ast\}$. 
For $\mathcal{O}\subset\subset \mathbb{Z}^d$, let $\mathcal{D}_\mathcal{O}$ denote the family of finite polynomials in the creation and annihilation operators $A_k^\sharp$, $k\in \mathcal{O}$ and let $\mathcal{D} \equiv \cup_{\mathcal{O}\subset\subset\mathbb{Z}^d}\mathcal{D}_\mathcal{O}$. 
%be an algebra generated by all the local algebras.
A particle number operator at $j\in\mathbb{Z}^d$ is defined by
\[
N_j\equiv A_j^\ast A_j
\]
\noindent\textbf{Remark} \\
The following relations will be useful in proving many claims involving creators/annihilators and particle number operators. For analytic function $h$,
(using a basis of eigenvectors for $N$), one can see that 
\begin{equation}\label{R1}\begin{split}
    A^{\ast}   h (N)  &=  h (N - 1)  A^{\ast}\quad\textrm{and} \quad
 		  h (N) A^{\ast} = A^{\ast}   h (N + 1) \\
     A  h (N-1) &=  h (N)  A\quad\textrm{and} \quad
 		 h (N+1)  A = A   h (N) \, .\
\end{split}
\end{equation}
For $\beta\in(0,\infty)$, let
\[
Z_o\equiv Tr_j \left(e^{-\beta N_j} \right)
\]
which is finite, positive  and independent of $j\in\mathbb{Z}^d$, and 
for $\Lambda\subset\subset\mathbb{Z}^d$ let
\[
\rho_{o,\Lambda}\equiv \prod_{j\in\Lambda}\frac{1}{Z_o} e^{-\beta N_j}
\]
\\
For $\Lambda\subset \mathbb{Z}^d$ define on $\mathcal{D}\cup \mathcal{A}$ the following linear maps
\[
\mathbb{E}_{o,\Lambda,s,p}(f) \equiv  Tr_\Lambda \left(\rho_{o,\Lambda}^{(1-s)/p}   f \rho_{o,\Lambda}^{s/p}  \right)
\]
with $0\leq s \leq 1$ and $p\in[1,\infty)$.
On $\mathcal{D}\cup \mathcal{A}$ we consider an infinite product state
\[
\omega_o(f)\equiv \lim_{\Lambda\to\mathbb{Z}^d} \mathbb{E}_{o,\Lambda,s}(f) 
%\otimes_{j\in\mathbb{Z}^d} \mathbb{E}_{o,j,s}(f)
\]
which is well defined on $\mathcal{D}\cup \mathcal{A}$.
%(with slight abuse of notation on the right hand side, understood as a number multiplying the identity operator).
The corresponding $\mathbb{L}_p(\omega_o)$, $p\in[1,\infty)$, functional is well defined on $\mathcal{D}$ (and $\cup \mathcal{A}$) by
\[
\|f\|_{\omega_o,p,s}^p = \lim_{\Lambda\to\mathbb{Z}^d}   Tr_\Lambda \left|\rho_{o,\Lambda}^{(1-s)/p}  f \rho_{o,\Lambda}^{s/p}\right|^p
\]
where $|g|\equiv \left(g^\ast g\right)^\frac12$. Replacing the limits by $\Lambda\to \mathcal{O}$ for any unbounded $\mathcal{O}\subset\mathbb{Z}^d$, we can introduce the corresponding state $\omega_{o,\mathcal{O}}$ and the corresponding $\mathbb{L}_p(\omega_{o,\mathcal{O}})$, $p\in[1,\infty)$, norms.
\\

\noindent\textbf{Interaction}\\
A potential $\Phi=\{(\Phi_\mathcal{O})_{\mathcal{O}\subset\subset \mathbb{Z}^d}\}$ of finite range $R\in(0,\infty)$ is a family of selfadjoint densely defined operators   $\Phi_\mathcal{O}\in \mathcal{A}_\mathcal{O}$ (or $\mathcal{D}_\mathcal{O}$), such that
\[\Phi_\mathcal{O}=0 \:\text{if}\: diam(\mathcal{O})\geq R.\]
The potential energy for a finite set $\Lambda\subset\subset\mathbb{Z}^d$ is defined by 
\[U_\Lambda=\sum_{\mathcal{O}\subset\Lambda }\Phi_\mathcal{O}.\]
We assume that the potential has the following thermodynamic stability property
\[
\begin{split}
0<Z_\Lambda \equiv Tr_\Lambda e^{-U_\Lambda}<\infty \\
    0< \left|\limsup_{\Lambda\to\mathbb{Z}^d} \frac{1}{|\Lambda|} \log Z_\Lambda \right| <\infty
\end{split}
\]
In the setup of this paper it is natural to consider the bounded multiparticle interaction in the form of polynomials in bounded operators $a_k^\sharp\equiv (1+\varepsilon N_k^{\frac12})^{-1} A_k^\sharp$, $k\in\mathbb{Z}^d$, $\varepsilon\in(0,\infty)$,
\[ \Phi_O\equiv\Phi_O(a_k^\sharp, k\in O).\]
\\

\noindent\textbf{A state and noncommutative $\mathbb{L}_p$ spaces}\\
On $\mathcal{D} \cup\mathcal{A} $, we define a state with respect to the interaction $U_\Lambda$
\[\omega^{(\Lambda)}(f)=\omega_o Tr_\Lambda  (\rho_\Lambda f)\] 
where the density matrix  is given by $\rho_\Lambda \equiv \frac{1}{Z_\Lambda}e^{-\beta U_\Lambda}$
and the associated modular dynamics is defined by 
\begin{equation}{\label{MO}}\alpha_{t,\Lambda} (B) %= e^{- {it\beta U_\Lambda}} {Be}^{{it\beta U_\Lambda}}
=\lim_{\tilde\Lambda\to\mathbb{Z}^d}\rho_{o,\tilde\Lambda\cap\Lambda^c}^{it}\rho_\Lambda^{it}B\rho_\Lambda^{-it}\rho_{o,\tilde\Lambda\cap\Lambda^c}^{-it}.\end{equation}
%which forms an automorphism group 
%on the von Neumann algebra $\mathcal{M}$ acting on the considered Hilbert space.
The corresponding scalar product is given as follows
\[\begin{split}
 \langle f, g \rangle_{\omega^{(\Lambda)}} %&= \frac{1}{Z_\Lambda} {Tr} (\rho_\Lambda^{1/2} f^{\ast}
%\rho_\Lambda^{1/2} g)\\
%&
=\omega^{(\Lambda)}((\alpha_{-i/2,\Lambda}(f))^\ast g)=\omega_\Lambda(f^\ast(\alpha_{-i/2,\Lambda }(g))).   
\end{split}\]
Under suitable conditions (discussed later, which hold e.g. for finite range  uniformly bounded weak multiparticle potential or quadratic ones) the following limits exist, (see \cite{BR} for bounded quantum spin systems),
\[
\omega (f) \equiv  \lim_{\Lambda\to\mathbb{Z}^d} \omega_\Lambda (f)
\]
\[
\alpha_t (f) \equiv  \lim_{\Lambda\to\mathbb{Z}^d} \alpha_{t ,\Lambda} (f)
\]
and consequently the following scalar product is well defined
\[\langle f, g \rangle_{\omega }=\lim_{\Lambda\to\mathbb{Z}^d}\langle f, g \rangle_{\omega_\Lambda} =\omega (f^\ast(\alpha_{-i/2}(g)))=\omega ((\alpha_{-i/4}(f))^\ast(\alpha_{-i/4}(g))) .
\]
%see e.g. \cite{BR2}. %\cite{Ru}, \cite{BR} \\
The corresponding $\mathbb{L}_p(\omega )$, for any $p\in[1,\infty)$, functional is well defined on $\mathcal{D}\cup \mathcal{A}$ by
\[
\|f\|_{\omega,p,s}^p = \lim_{\Lambda\to\mathbb{Z}^d}   Tr_\Lambda \left|\rho_{\Lambda}^{(1-s)/p}  f \rho_{\Lambda}^{s/p}\right|^p
\]
where $|g|\equiv \left(g^\ast g\right)^\frac12$.
Note that in particular case $p\in\mathbb{N}$ we have the following expression, see e.g. \cite{MZ},
\[
\|f\|_{\omega,p,1/2}^p  =\omega\left( \alpha(i /2p) (f)\alpha(3i/2p)(f)...\alpha((2p-1)i/2p)(f)\right) 
\]
where $\alpha( qi/2p)(f)\equiv\alpha_{qi/2p}(f)$, for a suitable class of operators $f$.\\

\noindent \textbf{Derivations}\\
A derivation $\delta_X$ in direction $X\in \mathcal{A} $ (resp. $\mathcal{D}$)  on $D(\delta_X)=\mathcal{A}$ (resp. $\mathcal{D}$)  is defined by
\[\delta_X(B)=i[X,B].\]

\begin{proposition} 
If $X^\ast X, X X^\ast\in\mathbb{L}_2(\omega)$, then the derivation is well defined on 
a domain
\[D(\delta_X)\supset \{B^\ast B,B B^\ast\in\mathbb{L}_2(\omega)\}.\]
In particular for any $X\in \mathcal{D}\cup \mathcal{A}$
\[\mathcal{D}\cup \mathcal{A} \subset D(\delta_X).
\]
The above condition holds if $X,\alpha_{\pm i/4}(X)\in \mathbb{L}_4(\omega)$ and then we have
$D(\delta_X)\supset \{B ,\alpha_{\pm i/4}(B)\in \mathbb{L}_4(\omega)\}$.\\   
\end{proposition}  
  \begin{proof}:
We note that 
\[
 \langle \delta_X(B), \delta_X(B) \rangle_{\omega }\leq 2\|XB\|_{\omega }^2 + 2\|BX\|_{\omega }^2
\]
If $ \tilde\omega(\cdot)=Tr (\tilde\rho\cdot)$, with a density matrix $\tilde \rho$, then using the definition of the norm and Cauchy-Schwartz inequality for trace we get
\[\begin{split}
   \|XB\|_{\tilde\omega }^2 = Tr\left(\tilde\rho^\frac12B^\ast X^\ast\tilde\rho^\frac12 XB\right) = Tr\left(B\tilde\rho^\frac12B^\ast\cdot  X^\ast\tilde\rho^\frac12 X\right)\\
   %\leq Tr\left(B\rho^\frac12B^\ast\cdot  (B\rho^\frac12B^\ast)^\ast\right)
 %  Tr\left(\left(X^\ast\rho^\frac12 X \right)^\ast\cdot  X^\ast\rho^\frac12 X\right)^\frac12
   \leq  \|B^\ast  B \|_{\tilde\omega }^2 \,
     \| X  X^\ast  \|_{\tilde\omega }^2 
\end{split}
 \]
 Similarly we have
 \[\begin{split}
  \|BX\|_{\tilde\omega }^2  = Tr\left(\tilde\rho^\frac12X^\ast B^\ast\tilde\rho^\frac12 BX \right) = Tr\left(B^\ast\tilde\rho^\frac12 B\cdot  X\tilde\rho^\frac12 X^\ast\right)\\
  % \leq Tr\left(B^\ast\rho^\frac12 B\cdot  (B^\ast\rho^\frac12 B)^\ast\right)
 %  Tr\left(\left(X\rho^\frac12 X^\ast \right)^\ast\cdot  X\rho^\frac12 X^\ast\right)^\frac12\\
   \leq  \|B B^\ast   \|_{\tilde\omega }^2 \,
     \|  X^\ast X  \|_{\tilde\omega }^2 
\end{split}
 \]
Next we remark that, with $\tilde\alpha_{\frac{i}4}(X) \equiv\tilde\rho^{-\frac14}X \tilde\rho^{  \frac14} $, we have
 \[\begin{split}
\|  X^\ast X  \|_{\tilde\omega }^2 =
Tr\left( \left( \tilde\rho^\frac18 X^\ast \tilde\rho^\frac18\right)
\left( \tilde\rho^\frac18 \tilde\rho^{-\frac14}X \tilde\rho^{  \frac14} \tilde\rho^\frac18\right)
 \left( \tilde\rho^\frac18 X^\ast \tilde\rho^\frac18\right)
 \left( \tilde\rho^\frac18 \tilde\rho^{-\frac14}X \tilde\rho^{  \frac14} \tilde\rho^\frac18\right)
\right)\\
=
Tr\left( \left( \tilde\rho^\frac18 X^\ast \tilde\rho^\frac18\right)
\left( \tilde\rho^\frac18 \tilde\alpha_{ \frac{i}4}(X)   \tilde\rho^\frac18\right)
 \left( \tilde\rho^\frac18 X^\ast \tilde\rho^\frac18\right)
 \left( \tilde\rho^\frac18 \tilde\alpha_{\frac{i}4}(X)   \tilde\rho^\frac18\right)
\right)\\
\leq \|X^\ast\|_{\tilde\omega,4 }^2 \|\tilde\alpha_{\frac{i}4}(X)\|_{\tilde\omega,4 }^2
\end{split}
 \]
 and similarly
 \[\begin{split}
\|  X X^\ast  \|_{\tilde\omega }^2  
\leq \|X \|_{\tilde\omega,4 }^2 \|\tilde\alpha_{\frac{i}4}(X^\ast)\|_{\tilde\omega,4 }^2
\end{split}
 \]
The rest follows by approximation of a state by normal states.
 \end{proof}

\subsection{Domain Issues}
We discuss briefly the domain questions when single point potential is given by $\Phi_{j}\equiv V(N_j)$ with unbounded %smooth 
real function $V$ (and the rest of the potential is bounded of finite range). 
 Consider the one point case with the state $\omega_V$ given by density matrix 
 $\exp\{-\beta V(N)\}/Z$. Then we notice that, using \eqref{R1}, we have 
\[\begin{split}
    \alpha_{is}(A^\ast)\equiv 
    e^{sV(N)}A^\ast e^{-sV(N)}&= e^{ s(V(N+1)-V(N))} A^\ast
    \\
    \alpha_{-is}(A)\equiv e^{-sV(N)} A\; e^{sV(N)}&= 
    e^{s(V(N+1)-V(N))} A 
\end{split}
\]
%\begin{equation}\label{R1}\begin{split}
%    A^{\ast}   h (N)  &=  h (N - 1)  A^{\ast}\quad\textrm{and} \quad
% 		  h (N) A^{\ast} = A^{\ast}   h (N + 1) \\
%     A  h (N-1) &=  h (N)  A\quad\textrm{and} \quad
% 		 h (N+1)  A = A   h (N) \, .\
%\end{split}\end{equation}
If $V$ is linear or sublinear the exponential multiplier is bounded.
However for $V$ growing faster than linearly, the exponential multiplier is an unbounded operator and $\alpha_{\pm is}(A^\sharp)$ may not be in $L_p(\omega_V)$
for some $s\neq 0$ and all sufficiently large $p\in[0,\infty)$.
In this case it is necessary to consider a suitable replacement for the space $\mathcal{D}$.
We propose to consider the following set 
\[\mathcal{D}_{\varepsilon,\Lambda}\equiv\{Fe^{-\sum_{j\in\Lambda}\varepsilon V(N_j)N_j^\delta }: F\in\mathcal{D}_\Lambda, \; \varepsilon,\delta\in(0,\infty)\} \]
and for infinite system
\[\mathcal{D}_{\varepsilon}\equiv \cup_{\Lambda\subset\subset\mathbb{Z}^d}\mathcal{D}_{\varepsilon,\Lambda}.
\]
Such set is dense in the closure of $\mathcal{D}$ with respect to $L_p$ norms.

%%%%%%%%%%%%%%%%%%%%%%%%%%%%%%%
\subsection{Adjoint Operators} %of a Derivation } %Subsection 2.2
We define left $L_X$, respectively right $R_X$, multiplication operator by an operator $X$ as follows
\[
L_X f\equiv X f,\qquad R_X f\equiv f X
\]
provided $Ran(f)\subset D(X)$ and $Ran(X)\subset D(f)$, respectively.

We have the following expressions for corresponding adjoints with respect to the scalar product.
\[ \begin{split}
\langle L_X (g), f\rangle &= Tr\left( \rho^\frac12 (Xg)^\ast \rho^\frac12 f\right)
= Tr\left( \rho^\frac12  g ^\ast \rho^\frac12 \left(\rho^{-\frac12}X^\ast \rho^\frac12\right) f\right)\\
&=    \langle g, L_{\alpha_{i/2}(X^\ast)} f\rangle
\end{split}
\]
and similarly we get
\[ 
\langle R_X (g), f\rangle  =  \langle g, R_{\alpha_{-i/2}(X^\ast)} f\rangle.
\]
Denoting the adjoint operation by $\star$, we summarise this as follows
\[
L_X ^\star = L_{\alpha_{i/2}(X^\ast)}, \qquad R_X^\star = R_{\alpha_{-i/2}(X^\ast)}.
\]
Next we discuss  adjoint of derivation with respect to a scalar product associated to a normal state $\omega(\cdot)\equiv Tr(\rho\cdot)$, with density operator $\rho$ with respect to a trace $Tr$.  
We have
\[\begin{split}
\langle \delta_{X} (f), g \rangle_{\omega} &\equiv 
{Tr} (\rho^{1/2} (\delta_{X} (f))^{\ast}\rho^{1/2} g)
={Tr} (\rho^{1/2} (i[X^{\ast},f^{\ast}])\rho^{1/2} g)\\
&={Tr} (\rho^{1/2} f^{\ast} \rho ^{1/2}
({g}\rho^{1/2} iX ^{\ast} \rho^{-1/2}-  \rho^{-1/2}  iX ^{\ast} \rho^{1/2} g))) 
%\\
%&=\langle f, \delta^{\star}_{X} (g) \rangle_{\omega}
\end{split}\]
Alternatively using the fact that
\[
\delta_{X} (f) = i ( L_X(f)- R_X(f)),
\]
we have
\[
\delta_{X}^\star   = -i \left( L_X^\star - R_X^\star \right) = 
i\left( R_{\alpha_{-i/2}(X^\ast)}-L_{\alpha_{i/2}(X^\ast)}\right).
\]
%Hence, (recalling  that $\alpha_{t}(g)\equiv \rho^{it}g\rho^{-it}$), we get
Summarising, we have the following property.

\begin{proposition}\label{AdjointPro}  For $X\in\mathcal{D}_\varepsilon$ 
\[\begin{split}
   \delta^{\star}_{X} (g) &= i\left( R_{\alpha_{-i/2}(X^\ast)}(g)-L_{\alpha_{i/2}(X^\ast)}(g)\right)\\
   &= i\left({g}\alpha_{-i/2}(X^{\ast}) - \alpha_{i/2}(X^{\ast})  g\right) \\
   &= - \delta_{\alpha_{-i/2}(X^{\ast})}(g) - i\left(\alpha_{i/2}(X^{\ast}) -\alpha_{-i/2}(X^{\ast})   \right)g \\
   &=  -  \delta_{\alpha_{i/2}(X^{\ast})}  (g) 
   - gi\left( \alpha_{i/2}(X^{\ast})  -\alpha_{-i/2}(X^{\ast})   
   \right) \\
  &=  -  \delta_{\frac12(\alpha_{i/2}(X^{\ast})+\alpha_{-i/2}(X^{\ast})}  (g) 
   + \left\{ \frac12 i\left( \alpha_{-i/2}(X^{\ast}) - \alpha_{i/2}(X^{\ast}) \right) , g\right\}
\end{split} \]
on a dense domain $\mathcal{D}(\delta^{\star}_{X})\supset\mathcal{D}$.
Moreover, for $f,g\in\mathcal{D}_\varepsilon$ , we have modified Leibnitz rule
\[ \begin{split}
  \delta^{\star}_{X} (fg) 
  %&= i\left({fg}\alpha_{-i/2}(X^{\ast}) - \alpha_{i/2}(X^{\ast})  fg\right) \\
&= \delta^{\star}_{X} (f) g  - f  \delta_{\alpha_{-i/2}(X^{\ast})}(g) 
\\
&= f \delta^{\star}_{X} (g) - \delta_{\alpha_{i/2}(X^{\ast})}(f)  g
\\
&=   \delta^{\star}_{X} (f) g  + f \delta^{\star}_{X} (g) 
-i f\left(\alpha_{-i/2}(X^{\ast})- \alpha_{i/2}(X^{\ast})  \right)g.
\end{split}
\]   

\end{proposition}   
\textbf{Remark}: We remark that in case when
\[
\alpha_{\pm \sfrac{i}{2}}(X)=e^{\pm \xi} X
\]
and so
\[
\alpha_{\mp \sfrac{i}{2}}(X^\ast)=e^{\pm \xi} X^\ast
\]
for some $\xi\in\mathbb{R}$, the above formulas simplify as follows
\[\begin{split}
   \delta^{\star}_{X} (g) &= i\left({g} (e^{\xi} X^{\ast}) -  (e^{-\xi}X^{\ast})  g\right) \\
   &= - e^{\xi}\delta_{ X^{\ast} }(g) +2i\sinh(\xi)  X^{\ast}  g \\
   &=  -  e^{-\xi}\delta_{X^{\ast}} (g) 
   + 2i\sinh(\xi) g X^{\ast} 
  \\
  &=  -  \cosh(\xi)\delta_{ X^{\ast}}  (g) 
   + i\sinh(\xi)\left\{  X^{\ast}  , g\right\}
\end{split} \]
and for the modified Leibnitz rule as follows
\[ \begin{split}
  \delta^{\star}_{X} (fg) 
&= \delta^{\star}_{X} (f) g  - e^{\xi}f  \delta_{X^{\ast}}(g) \\
&= f \delta^{\star}_{X} (g) - e^{-\xi}\delta_{X^{\ast}}(f)  g .
\end{split}
\]   
On the algebra associated to $\xi=0$, as e.g. algebra of functions of $N$ in case of quantum harmonic oscillator, we recover perfect classical formulas.

\section{Quasi Invariance and Unitary Group Representations} %Section 3

In this section we discuss briefly unitary representations of some groups. One of the motivations for that is the fact that the generators of such representations can be used in constructions of dissipative dynamics discussed later on.\\
Suppose, $\theta, \tau\in\mathbb{C}$ satisfy $|\tau|^2-|\theta|^2=1$. Then the operators
\[
a\equiv a(\tau,\theta)\equiv \tau A+\theta A^\ast, \qquad a^\ast\equiv a^\ast(\tau,\theta)\equiv \bar\tau A^\ast +\bar\theta A 
\]
satisfy CCR
\[[a,a^\ast]  =id.\]
The transformation of this type is known in the mathematical/theoretical physics literature under the name of \href{https://en.wikipedia.org/wiki/Bogoliubov_transformation}{Bogolubov transformations}.\\
We can consider here a class of transformations which includes the Lorenz group acting on two dimensional CCR vectors via matrices, %$t\in\mathbb{R}$
\[
\mathbb{R}\ni t \longmapsto
\begin{pmatrix}
cosh(t) & sinh(t)  \\
sinh(t) & cosh(t)
\end{pmatrix}
\]
(which can be extend to complex parameter).
For any polynomial function $f$ and interaction energy $U$ for which $Tr \left( f(A,A^\ast) e^{-U(A,A^\ast)}\right)$ is well defined, we have
\[
Tr\left( f(a,a^\ast) e^{-U(a,a^\ast)}\right)=Tr \left(f(a,a^\ast) e^{-U(a,a^\ast)} \right).
\]
For a differentiable function $\mathbb{R}\ni s\to \{(\tau(s), \theta(s)):|\tau(s)|^2= |\theta(s)|^2+1\}$, we define  
\[a_s\equiv a(\tau(s), \theta(s)), \; a_s^\ast\equiv a(\tau(s), \theta(s))^\ast,\]
with the initial condition
$a_0=A, a_0^\ast=A^\ast$.
Given initial density $\rho\equiv \frac1Z e^{-U(A,A^\ast)}$, we define transformed density as follows
\[
\rho_s\equiv \frac1{Z_s} e^{-U(a_s,a_s^\ast)}
\]
with normalisation factor $Z_s =Tr e^{-U(a_s,a_s^\ast)}=Tre^{-U(A,A^\ast)}$.
Then we have the following result.
\begin{theorem}
The following formula defines a unitary group representation in $\mathbb{L}_2(\rho)$.
\[
V_s(f(A,A^\ast)) \equiv \rho^{-\frac14}\rho_s^{\frac14} f(a_s,a_s^\ast)
\rho_s^{\frac14}\rho^{-\frac14}
\]
The generator of the group on polynomials is given by
\[
\partial_sV_s(f)(A,A^\ast)_{|s=0} = \partial_s f(a_s,a_s^\ast)_{|s=0} -\left\{ \frac14\int_0^1 d\lambda \; \rho^{\frac{\lambda}4} \partial_s U(a_s,a_s^\ast)_{|s=0}\rho^{\frac{1-\lambda}4}, f(A,A^\ast)-\frac12 \frac{Z'}{Z}f(A,A^\ast)\right\}.
\]
where curly bracket denote anticommutator and $Z'\equiv (\partial_s{Z_s})_{|s=0}$.   
\end{theorem}

To get to the representation of Lorenz group in higher space dimension
one needs to consider higher order quantisation of space-time in a form of product space of many independent harmonic oscillators $(A_i,A_i^\ast)_{i=0,1..,n}$ in which we can consider  the following CCR representation
\[(A_i,A_i^\ast)_{i=1,..,n}\mapsto S(\tau,\mathbf{x})\equiv\tau \frac{1}{\sqrt{n}}\sum_{i=1,..,n} A_i+\sum_{i=1,..,n} x_i A_i^\ast.\]
[So in infinite dimensional limit time component is given by a "Gaussian random variable".]
Then we have the following representation of the Minkowski scalar product
\[[S,S^\ast]=|\tau|^2 -|\mathbf{x}|^2.\]

\textbf{Remark} {\textit{ Given $n\in\mathbb{N}$, $n>1$, CCR pairs
$(A_i,A_i^\ast)_{i=0,1..,n}$, we can define 
\[\begin{split}
A_i(\boldsymbol{\gamma}, \boldsymbol{\kappa})\equiv \sum_{j=1,..,n} \left(\gamma_{ij} A_j + \kappa_{ij}A_j^\ast\right) \\
 A_i^\ast(\boldsymbol{\gamma}, \boldsymbol{\kappa})\equiv 
 \sum_{j=1,..,n} \left(\bar\gamma_{ij} A_j^\ast + \bar\kappa_{ij}A_j\right)
\end{split}
\]
with the CCR condition 
$[A_i(\boldsymbol{\gamma}, \boldsymbol{\kappa}),A_i^\ast(\boldsymbol{\gamma}, \boldsymbol{\kappa})]=id$ given by
\[%\begin{split}
%[A_i(\boldsymbol{\gamma}, \boldsymbol{\kappa}),A_i^\ast(\boldsymbol{\gamma}, \boldsymbol{\kappa})] =[ \sum_{j=1,..,n} \left(\gamma_{ij} A_j + \kappa_{ij}A_j^\ast\right), 
 %\sum_{k=1,..,n} \left(\bar\gamma_{ik} A_k^\ast + \bar\kappa_{ik}A_k\right)]\\
 %= 
 \sum_{j=1,..,n} |\gamma_{ij} |^2 - \sum_{j=1,..,n} |\kappa_{ij} |^2 = 1
%\end{split}
\]
Hence we can introduce a group of linear transformations \[\begin{split}
\mathbf{T}\equiv(T,\tilde T): (\boldsymbol{\gamma}, \boldsymbol{\kappa})\mapsto 
\mathbf{T}(\boldsymbol{\gamma}, \boldsymbol{\kappa})\equiv (T(\boldsymbol{\gamma}),\tilde T(\boldsymbol{\kappa}))  \\
 (\mathbf{T}(\boldsymbol{\gamma}))_i\equiv \sum_{j=1,..,n} \left(T_{ij}\gamma_{ij}\right), \qquad 
 (\mathbf{T}(\boldsymbol{\kappa}))_i \equiv \sum_{j=1,..,n} \left(
 \tilde T_{ij}\kappa_{ij}\right) 
 \end{split}
 \]
preserving the CCR condition, i.e. satisfying
\[
 \sum_{j=1,..,n} | T(\boldsymbol{\gamma})_{ij} |^2 - \sum_{j=1,..,n} |\tilde T(\boldsymbol{\gamma})_{ij} )|^2 = 1.
\]
%
% SO(3,1) is the group that leaves the quadratic form
%dx^2+dy^2+dz^2‒dw^2 invariant
% \url{https://physics.stackexchange.com/questions/366317/isometry-group-of-minkowski-space}
%
}}
We can use the idea of the quasi-invariance of a state to obtain unitary representations of this extended group in a similar fashion.\\
More representations of groups in multicomponent systems are considered later on, (see also \cite{MZ}). 
%
%%%%%%%%%%%%%%%%%%%%%%%%%%%%%%%
\section{Modular Dynamics and Finite Speed of Propagation of Information}
It is sufficient to define modular dynamics for operators for creation  and anihilation operators. We construct the modular dynamics first on bounded operators defined by polynomials of the following modified creation  and anihilation operators
\[a_j^\sharp\equiv \frac{1}{1+\epsilon N_j^\frac12}A_j^\sharp.\]

For $\Lambda\subset\subset\mathbb{Z}^d$, $\alpha_{t ,\Lambda}$ is well defined  and for a given $j\in\Lambda$ and $k\notin\Lambda$,
%$f\in\mathcal{D}_{\varepsilon,\Lambda}$ 
we have
%\[
% \frac{d}{dt}\alpha_{t ,\Lambda} (\frac{1}{N_j^\frac12}A_j^\sharp)=   -ie^{-it U_\Lambda} [U_\Lambda,\frac{1}{N_j^\frac12}A_j^\sharp]e^{ it U_\Lambda}
%\]
\[\begin{split}
    \alpha_{t ,\Lambda\cup\{k\} }(a_j^\sharp)-\alpha_{t ,\Lambda}(a_j^\sharp) &=  -\int_0^t ds  \frac{d}{ds}\alpha_{s ,\Lambda }\left(\alpha_{t-s ,\Lambda\cup\{k\}  }(a_j^\sharp)\right)\\
    &=  \int_0^t                                                              \sum_{O\ni k\atop O\neq\{k\}}\alpha_{s ,\Lambda  }\delta_{\Phi_O}\left(\alpha_{t-s ,\Lambda\cup\{k\}  }(a_j^\sharp)\right) ds
\end{split}
\]
Assuming $\Phi_O$ are bounded if $O$ is not one point set, this implies
the following bound involving operator norm
\[\begin{split}
   \| \alpha_{t ,\Lambda\cup\{k\} }(a_j^\sharp)-\alpha_{t ,\Lambda}(a_j^\sharp)\| \leq  \int_0^t   \sum_{O\ni k\atop O\neq\{k\}}\|\delta_{\Phi_O}\left(\alpha_{s ,\Lambda\cup\{k\}  }(a_j^\sharp)\right)\| ds
\end{split}
\]
To estimate the right hand side we use the following result where
\[
 c_\Phi\equiv 2 sup_{|diam (O)|\leq 2R}\sum_{O'\subset O_R }^{\tilde{}}\|\Phi_{O'}\| , 
 \]
where the summation with $\tilde{}$ runs over sets different than one point sets and $O_R$ denotes a set of points with distance from $O$ bounded by $R$.\\

\begin{theorem}(Finite speed of propagation of information estimate)\\
Assume the potential is of finite range $R\in(0,\infty)$ and 
$c_\Phi<\infty$.
There exist  constants $D,C,m\in\mathbb{R}^+$ such that for any $j\in\Lambda\subset\subset\mathbb{Z}^d$ and any $t\in\mathbb{R}^+$ we have
\[\|\delta_{\Phi_O}\left(\alpha_{t ,\Lambda   }(a_j^\sharp)\right)\| \leq De^{Ct-md(O,j)}\]
where $\Phi_O$ is a bounded part of potential localised in a set of size $2R$. 
The estimate remains valid for $t\in\mathbb{C}$, $|\mathfrak{Im}({t})|\leq 1$, with $\mathfrak{Re} (t)$ on the right hand side, provided
one point interaction $V(N_k)$ is at most linear or $a_j\in\mathcal{D}_{\varepsilon}$.
\end{theorem}
 \begin{proof}  If the interaction if of range $R\in(0,\infty)$, we note that when $dist(O,j)>2R$ and $j\in\Lambda\setminus O_R$, we have 
\[\delta_{\Phi_O}\left(\alpha_{s ,\Lambda\setminus  O_R  }(a_j^\sharp)\right)=0,\]
where  $O_R\equiv\{l: dist(l,O)\leq 2R\}$.
Hence
\[ 
   \|\delta_{\Phi_O}\left(\alpha_{t ,\Lambda }(a_j^\sharp)\right)\| = \|\delta_{\Phi_O}\left(\alpha_{t ,\Lambda  }(a_j^\sharp) - \alpha_{t ,\Lambda\setminus O_R}(a_j^\sharp)\right)\|
   \]
We have 
\[\begin{split}
   \alpha_{t ,\Lambda  }(a_j^\sharp) - \alpha_{t ,\Lambda\setminus O_R}(a_j^\sharp) = -\int_0^t ds \frac{d}{ds}
  \alpha_{t-s ,\Lambda\setminus O_R}(\alpha_{s ,\Lambda  }(a_j^\sharp) ) \\
  -\int_0^t ds  
  \alpha_{t-s ,\Lambda\setminus O_R}\left(\sum_{O'\subset O_R }^{\tilde{}}\delta_{\Phi_{O'}}\left( \alpha_{s ,\Lambda }(a_j^\sharp) \right)\right)
\end{split}
\]
where  $\tilde{}$ over the sum indicates summation over the bounded part of the potential, since $\alpha_{\tau,\Lambda'}$ is defined outside the set $\Lambda'$ by using the one point potential only and this is cancelled out in the derivation with respect to $s$. (Here we may need to consider first bounded approximation of the one point potential.) Hence we get
  \[\begin{split}
   \|\delta_{\Phi_O}\left(\alpha_{t ,\Lambda }(a_j^\sharp)\right)\| \leq 2\|\delta_{\Phi_O}\left(\int_0^t ds \alpha_{t-s ,\Lambda\setminus O_R } \sum_{O'\subset O_R }^{\tilde{}}\delta_{\Phi_{O'}}\left( \alpha_{s ,\Lambda }(a_j^\sharp) \right)\right)\|\\
   \leq 2\|\Phi_O\| \sum_{O'\subset O_R }^{\tilde{}}
   \int_0^t ds \| \delta_{\Phi_{O'}} \left(\alpha_{s ,\Lambda}(a_j^\sharp)\right) \|
\end{split}
 \]
 Since by our assumption
 \[
 c_\Phi\equiv 2 sup_{|diam (O)|\leq 2R}\sum_{O'\subset O_R }^{\tilde{}}\|\Phi_{O'}\|<\infty, 
 \]
we get the following bound
\[ 
   \|\delta_{\Phi_O}\left(\alpha_{t ,\Lambda }(a_j^\sharp)\right)\|  
   \leq  c_\Phi\sum_{O'\subset O_R }^{\tilde{}}
   \int_0^t ds \| \delta_{\Phi_{O'}} \left(\alpha_{s ,\Lambda}(a_j^\sharp)\right) \|
 \]
We iterate this bound , each time getting finite constant multiplier $c_{\Phi}$, a finite sum of at most $2^{(2R)^d} $ terms (which is a maximal number of $O'\subseteq O_R$) and iterated integral.
The iteration is terminated when one of the sets ${O'}_R$ contains $j$. The minimal number of steps to reach $j$ is at least $n\equiv [dist(O,j)/2R]$ (where the square bracket indicates the integer part). Thus, with $C\equiv 2^{(2R)^d} c_\Phi $, we obtain the following bound 
\[ 
   \|\delta_{\Phi_O}\left(\alpha_{t ,\Lambda }(a_j^\sharp)\right)\|  
   \leq  C^n  \frac{t^n}{n!} e^{Ct}\|a_j^\sharp \|
   \]
   At this point one needs to use additional assumptions on $a_j^\sharp$.
Using Stirling-de Moivre bound  
\[n!\geq n^{n+1/2} e^{-n}\]
we get 
\[ 
   \|\delta_{\Phi_O}\left(\alpha_{t ,\Lambda }(a_j^\sharp)\right)\|  
   \leq   \exp\{n(\log (Cet) -\log n)\} e^{Ct}\|a_j^\sharp \|
   \]
which for any fixed $t\in\mathbb{R}^+$ and $n \geq Cet \exp (m(2R)^d)$,  yields
\[ 
   \|\delta_{\Phi_O}\left(\alpha_{t ,\Lambda }(a_j^\sharp)\right)\|  
   \leq    e^{Ct-mn}\|a_j^\sharp \|
   \]
    \end{proof}
    
   Our result generalises that of \cite{LR} (\cite{Mat2}) who consider similar result for a system of bounded spin on a lattice.
   
\section{Convergence of $\mathbb{L}_p$ Norms}

In this section we provide a brief discussion of convergence of the sequence of states
$\omega^{(\Lambda)}$, $\Lambda\subset\subset\mathbb{Z}^d$ and corresponding $\mathbb{L}_p$ norms.
Let  $\Lambda_0\subset\Lambda$ and $\partial_{2R}\Lambda_0\equiv \{k\notin\Lambda_0: dist(k,\Lambda_0)\leq 2R\}$. We will use the following interpolation of the the potential \[\Phi_s\equiv \{\Phi_O, O\subset \Lambda\setminus\partial_{2R}\Lambda_0, s\Phi_{O'},  O'\cap \partial_{2R}\Lambda_0\neq \emptyset \}_{\sim},\]
where $\sim$ signifies that we exclude one point potential,  and denote by $\rho_{\Lambda,s}\equiv\rho_{\Lambda,\Lambda_0,s}$ the corresponding density matrix localised in $\Lambda$.
Note that for $s=0$, we have \[\rho_{\Lambda,s=0}=\rho_{\Lambda\setminus\partial_{2R}\Lambda_0}\rho_{\Lambda_0}\rho_{o,\partial_{2R}\Lambda_0}\]
where the density matrices on the right hand side commute and the corresponding state is a product state.
For $p\in \mathbb{N}$ and $f$ localised in $\Lambda_0$, we have
\[
Tr_\Lambda \rho_{\Lambda,s=0}f=Tr_{\Lambda_0} \rho_{\Lambda_0}f
\]
and 
\[
\|f\|_{p,\Lambda,s=0}= \|f\|_{p,\Lambda_0}.
\]
Next we have
\[\begin{split}
Tr_{\Lambda }\left(\left(\rho_{\Lambda }\right)^\frac{1}{4p}f^\ast\left( \rho_{\Lambda }\right)^\frac{2}{4p}f \left( \rho_{\Lambda }\right)^\frac{1}{4p}\right)^{2p}
&-Tr_{\Lambda }\left(\left(\rho_{\Lambda,s=0}\right)^\frac{1}{4p}f^\ast\left(\rho_{\Lambda,s=0}\right)^\frac{2}{4p}f \left(\rho_{\Lambda,s=0}\right)^\frac{1}{4p}\right)^{2p}  \\
&=\int_0^1ds \frac{d}{ds} Tr_{\Lambda}\left(\left(\rho_{\Lambda ,s}\right)^\frac{1}{4p}f^\ast\left( \rho_{\Lambda ,s}\right)^\frac{2}{4p}f \left( \rho_{\Lambda ,s}\right)^\frac{1}{4p}\right)^{2p}
\end{split}
\]
Next we note that
\[
\frac{d}{ds} \left(\rho_{\Lambda ,s}\right)^\frac{1}{4p} = \left(\frac{d}{ds} \exp\{-\frac{1}{4p}U_{\Lambda }(\Phi_s)\}\right)\frac{1}{Z_{\Lambda,s}^\frac{1}{4p}}
-\frac{1}{4p}\left(\rho_{\Lambda ,s}\right)^\frac{1}{4p}\cdot
\frac{1}{Z_{\Lambda,s}}\frac{d}{ds}Z_{\Lambda,s}
\]
For the derivative of the first factor on the right hand side we have
\[\begin{split}
\frac{d}{ds} \exp\{-\frac{1}{4p}U_{\Lambda }(\Phi_s)\}=\lim_{h\to 0}\frac{1}{h}\left(\exp\{-\frac{1}{4p}U_{\Lambda }(\Phi_{s+h})\}-\exp\{-\frac{1}{4p}U_{\Lambda }(\Phi_{s})\}\right)\\
=\lim_{h\to 0}\frac{1}{h}\int_0^1 d\tau \frac{d}{d\tau} \exp\{-\frac{\tau}{4p}U_{\Lambda }(\Phi_{s+h})\} \exp\{-\frac{1-\tau}{4p}U_{\Lambda }(\Phi_{s})\}\\
=-\frac{1}{4p}\int_0^1 d\tau   \exp\{-\frac{\tau}{4p}U_{\Lambda }(\Phi_{s+h})\} \lim_{h\to 0}\frac{1}{h}\left( U_{\Lambda }(\Phi_{s+h})- U_{\Lambda }(\Phi_{s})\right)\exp\{-\frac{1-\tau}{4p}U_{\Lambda }(\Phi_{s})\}\\
= -\frac{1}{4p} \sum_{O\subset \partial_{2R}\Lambda }^\sim \int_0^1 d\tau   \left(\rho_{\Lambda ,s}\right)^\frac{\tau}{4p}  (\Phi_O) \left(\rho_{\Lambda ,s}\right)^{-\frac{\tau}{4p} } \rho_{\Lambda ,s}^\frac{1}{4p}\\
\equiv %{\color{red}
-\frac{1}{4p} \sum_{O\subset \partial_{2R}\Lambda }^\sim \int_0^1 d\tau   \left(\alpha_{\Lambda ,s}(-i\frac{\tau}{4p})  (\Phi_O)  \right) \rho_{\Lambda ,s}^\frac{1}{4p} %}
\end{split}
\]
where $\alpha_{\Lambda ,s}(-i\frac{\tau}{4p})$ denotes the automorphism corresponding to $\Phi_{s}$ at time $-i\frac{\tau}{4p}$, and similarly for the second power of the density. Hence, we also get
\[\begin{split}
   \left| \frac{1}{Z_{\Lambda,s}}\frac{d}{ds}Z_{\Lambda,s}\right| &= \left|-\sum_{O\subset \partial_{2R}\Lambda_0}^\sim \int_0^1 d\tau  \frac{1}{Z_{\Lambda,s}} Tr_\Lambda\left(\rho_{\Lambda ,s}\right)^{\tau } (\Phi_O) \left(\rho_{\Lambda ,s}\right)^{- \tau } \rho_{\Lambda ,s} \right|
\\
&= \left|-\sum_{O\subset \partial_{2R}\Lambda_0}^\sim \int_0^1 d\tau    \omega_{\Lambda,s}  \left(\Phi_O\right)\right|
\leq \sum_{O\subset \partial_{2R}\Lambda_0}^\sim \|\Phi_O\|
\end{split}
\]
 
Using the above, we obtain 
\[\begin{split}
&\frac{d}{ds} Tr_{\Lambda }\left(\left(\rho_{\Lambda ,s}\right)^\frac{1}{4p}f^\ast\left( \rho_{\Lambda ,s}\right)^\frac{2}{4p}f \left( \rho_{\Lambda ,s}\right)^\frac{1}{4p}\right)^{2p} \leq \\
&\qquad\qquad\qquad\qquad\qquad  \leq Tr_{\Lambda }\left(\left(\rho_{\Lambda ,s}\right)^\frac{1}{4p}f^\ast\left( \rho_{\Lambda ,s}\right)^\frac{2}{4p}f \left( \rho_{\Lambda ,s}\right)^\frac{1}{4p}\right)^{2p} 
\left( \sum_{O\subset \partial_{2R}\Lambda_0}^\sim \left( \|\Phi_O\| +%{\color{red}
\sup_{\Lambda,s,\tau\in[0,1]} 
\|\alpha_{\Lambda ,s}(-i\frac{\tau}{4p})  (\Phi_O) \|%}
\right)\right)
\end{split}
\]
and hence we arrive at the following bound
\begin{equation} \label{Compactness}
    \| f\|_{\Lambda,p}\leq \| f\|_{\Lambda_0,p}e^{C|\partial_{2R}\Lambda_0|}
\end{equation} 
with a constant 
\[
C\leq \sum_{O\subset \partial_{2R}\Lambda_0}^\sim \left( \|\Phi_O\| +%{\color{red}
\sup_{\Lambda,s,\tau\in[0,1]} \|\alpha_{\Lambda ,s}(-i\frac{\tau}{4p})  (\Phi_O) \|%}
\right)
\]
which according to the considerations of the previous subsection is finite
under suitable assumptions on bounded part of the potential.\\

The above bounds \eqref{Compactness} provide compactness of the set of states
$\omega_{\Lambda}$, $\Lambda\subset\subset\mathbb{Z}^d$,
and we have the following possibility of defining  $\mathbb{L}_p(\omega)$ norms associated to a state $\omega\equiv \lim_{\Lambda_k\to\mathbb{Z}^d}\omega_{\Lambda_k}$ for some subsequence $\Lambda_k\subset\Lambda_{k+1}$,
\[
\|f\|_{\omega,p}\equiv \limsup_{\Lambda\to\mathbb{Z}^d} \|f\|_{\omega_{\Lambda},p}
\]
Under additional assumptions on the interactions it is possible to use the ideas utilised above to prove convergence of the $\|f\|_{\omega_{\Lambda},p}$ as $\Lambda\to \infty$.\\
Note that for positive $f$ the symmetric $\mathbb{L}_{\omega_\Lambda,1}$ coincides with $\omega_\Lambda(f)$, so the problem of convergence is the same for both. On the other hand for 
the symmetric $\mathbb{L}_{\omega_\Lambda,2}$ , given the convergence of the sequence of state and 
the modular operator, we get convergence for corresponding norms.
Given $\mathbb{L}_{\omega,1}$ and $\mathbb{L}_{\omega,2}$ one can use interpolation theory to get all the intermediate norms and spaces and then by duality one can define the norms and $\mathbb{L}_{\omega,p}$ spaces for $p\in(2,\infty)$.\\

\section{Noncommutative Dirichlet forms and Markov Generators}
To give the definition of Dirichlet forms (\cite{Cip3,CZ,P1,MZ,SQV}), we first describe the admissible function:
\begin{definition}{\label{admissible}}
An analytic function $\eta:D\rightarrow\mathbb{C}$ on a domain $D$ containing the strip $Im\:z\in [-1/4,1/4]$ is said to be admissible function if the following holds:
\begin{enumerate}
    \item $\eta(t)\geq 0$ for $t\in\mathbb{R}$,
    \item $\eta(t+i/4)+\eta(t-i/4)\geq 0$ for $t\in\mathbb{R}$,
    \item there exists $M>0$ and $p>1$ such that the bound 
    \[|\eta(t+is)|\leq M(1+|t|)^{-p}\]
    holds uniformly in $s\in [-1/4,1/4]$.
\end{enumerate}
\end{definition}
A nice example of such function is as a gaussian smoothing of the following
\begin{equation}\label{Ch}
    \eta(t)\equiv \frac{e^{i \kappa t}}{cosh(2n\pi t)}.
\end{equation}

For a collection of other examples see \cite{CZ}, \cite{P1}.

%%%%%

%%%%

Given a family of local elements $X_j\in\mathcal{D}$, $j\in\mathbb{Z}^d$, we have discussed above, we can define derivations in directions $\alpha_t(X_j)$ and, on a dense domain containing $\mathcal{D}_\varepsilon$,   define the following Dirichlet form in $\mathbb{L}_{\omega,2}$
\[\mathcal{E}_j(f)\equiv\int_{\mathbb{R}} \left(\nu_j\langle \delta_{\alpha_t(X_j) } (f), \delta_{\alpha_t(X_j)} (f)\rangle_{\omega}+\mu_j\langle \delta_{\alpha_t(X_j^\ast) } (f), \delta_{\alpha_t(X_j^\ast)} (f)\rangle_{\omega}\right) \eta(t)dt
\]
%\[D(\mathcal{E}_\Lambda)=\{f\in \mathcal{H}_\Lambda:\mathcal{E}_\Lambda(f)<\infty\}\]
for %{\color{red}
$\nu_j,\mu_j$%} 
and a finite $\Lambda\subset\mathbb{Z}^d$
\[
\mathcal{E}_\Lambda(f)\equiv \sum_{j\in \Lambda} \mathcal{E}_j(f)\]
Some other Dirichlet form can be defined %(with modular automorphism  $\alpha_{\omega}$ corresponding to the state $\omega$)
, see e.g.\cite{CM}, as follows
\[
\tilde{\mathcal{E}}_\Lambda(f)\equiv \sum_{j\in \Lambda} \left(\nu_j\langle  \delta_{E_j } (f), \delta_{E_j} (f)\rangle_{\omega}+\mu_j\langle \delta_{E_j^\ast } (f), \delta_{E_j^\ast} (f)\rangle_{\omega}\right)
\tag{$\tilde{\mathcal{E}}$}
\]
where $E_j$ are the eigenvectors of the modular operator $\alpha_{\omega}(\pm\tfrac{i}2)$, 
associated to the state $\omega$ at time $\pm\tfrac{i}2$, such that

\[%{\color{red}
\alpha_{\omega}(\pm\tfrac{i}2)(E_j)%}
=e^{\pm\xi_j} E_j,\] 
for some $\xi_j\in\mathbb{R}$ , and $\nu_j,\mu_j\in(0,\infty)$.

\begin{comment}

One approach to obtain such eigenvectors was recently discussed in 
%by Cipriani and Zegarlinski 
\cite{CZ}, where 
\begin{equation}{\label{EV}}
    E_j=\int_\mathbb{R} \alpha_t(X_j)\eta(t)dt.
\end{equation}
with a suitable function $\eta$. %being an analytic function in a strip satisfying suitable periodicity conditions 
%\[\eta(t\pm i/2)=e^{\pm \xi}\eta(t)\]
%for some $\xi\in \mathbb{R}$. 
%(For example using \eqref{Ch} with various $n$ even and $cosh$ replace by a suitable function of $cosh$.)\\
One can show that the following relation is true 
\begin{proposition} \label{Pro 6.1}
{\color{blue}\[\tilde{\mathcal{E}}_\Lambda(f)\leq \mathcal{E}_\Lambda(f) .\]}
\end{proposition}
\end{comment}

Now, the Markov generator corresponding to the first, respectively the second, Dirichlet form is formally give by
\[\langle f,-\tilde{\mathfrak{L}}_\Lambda f\rangle=\tilde{\mathcal{E}}_\Lambda(f) 
,\quad \text{resp}.\quad 
\langle f,-\mathfrak{L}_\Lambda f\rangle=\mathcal{E}_\Lambda(f).\]
Using operation $\star$ of taking the adjoint (with respect to the scalar product), we have
\[
-\mathfrak{L}_j (f) = \int   \left(\nu_j \delta^\star_{\alpha_t(X_j)}\delta_{\alpha_t(X_j)}(f) +  \mu_j \delta^\star_{\alpha_t(X_j^\ast)}\delta_{\alpha_t(X_j^\ast)}(f) \right)\eta(t) dt .
%{\color{blue} \kappa,\mu?}
\]

%{\color{blue} Working Remark: Add formula $\tilde{\mathfrac{L}}$ and corresponding $\tilde{\Gamma}_{1,\Lambda}$}\\
For the case of infinite number of eigenvectors 
\[
X_j=\left(\int\alpha_t(B_{j,l})\eta(t)dt\right)_{l=1,..,K}, \quad j\in\mathbb{Z}^d, %{\color{blue}B_{j,l}, l\:finite?}
\]
the tilded generator is as follows
\[ 
-\tilde{\mathfrak{L}} (f^\ast f) =    \sum_{j\in\mathbb{Z}^d} \left(\kappa\delta^\star_{X_j}\delta_{X_j}(f^\ast f) +   \mu \delta^\star_{X_j^\ast}\delta_{ X_j^\ast}(f^\ast f)\right).  
 \]
The corresponding Markovian form is defined by
\[\Gamma_{1,\Lambda}(f)=\frac{1}{2}\left(\mathfrak{L}_\Lambda(f^\ast f)-f^\ast\mathfrak{L}_\Lambda( f)-\mathfrak{L}_\Lambda(f^\ast )f\right)\]
and in our setup is given by %\ref{AdjointPro}
\[\begin{split}2\Gamma_{1,\Lambda}(f)&=  \sum_{j\in\Lambda} \left\{\quad \int   \left(\nu_j\delta^\star_{\alpha_t(X_j)}\delta_{\alpha_t(X_j)}(f^\ast f) +  \mu_j \delta^\star_{\alpha_t(X_j^\ast)}\delta_{\alpha_t(X_j^\ast)}(f^\ast f) \right)\eta(t) dt \right. \\
&
-f^\ast\left(\int   \left(\nu_j\delta^\star_{\alpha_t(X_j)}\delta_{\alpha_t(X_j)}(f) +   \mu_j\delta^\star_{\alpha_t(X_j^\ast)}\delta_{\alpha_t(X_j^\ast)}(f) \right)\eta(t) dt)\right) 
\\
&\left. -\int   \left(\nu_j\delta^\star_{\alpha_t(X_j)}\delta_{\alpha_t(X_j)}(f^\ast) +   \mu_j\delta^\star_{\alpha_t(X_j^\ast)}\delta_{\alpha_t(X_j^\ast)}(f^\ast) \right)\eta(t) dt f \quad \right\} %{\color{blue}\kappa,\mu?}
\end{split}\]
and similarly in the tilded case.\\

Using Proposition \ref{AdjointPro} after lengthy computations (provided for the reader convenience in an Appendix), one gets (for every $j$) a result of the following form.
\begin{proposition}\label{Gamma1Proposition}
   \[  
      \Gamma_{1}(f) =-\frac{1}{2}\int \left( \left|\delta_{\alpha_{t-i/4}(X^\ast)}(f)\right|^2+
    \left|\delta_{\alpha_{t-i/4}(X)}(f))\right|^2 \right) \left(\eta(t+i/4) 
   +\eta({t-i/4})\right) dt
    \]
    and respectively   
\[
\tilde{\Gamma}_1(f) =  \sum_{j\in\mathbb{Z}^d} \left(\kappa e^{-\xi}   |\delta_{X_j}(f) |^2  
+   \mu  e^\xi| \delta_{X_j^{\ast}} (f) |^2 \right)
\]
In the infinite dimensional case to secure the dense domain it is necessary to have the finite speed of propagation of information (as even for localised $B_{j,l}$, in general the eigenvectors are not localised).
\end{proposition}

\section{Examples of models} % with Quadratic interactions}
In this section we provide few explicit examples of models with interesting phenomena.
\subsection{Mean Field Models}
\begin{example}(Mean Field Model){\label{ex1}} %Example 7.1
With $[X_j,X_k^\ast]=\delta_{j,k}$, define
\[
\mathbf{X}_\Lambda \equiv \frac1{\sqrt{|\Lambda|}}\sum_{k\in\Lambda} X_k
\]
Then the following CCR holds 
 \[\left[\mathbf{X}_\Lambda, \mathbf{X}_\Lambda^\ast \right] %=\left(\frac1{\sqrt{|\Lambda|}}\right)^2 \sum_{k\in\Lambda} [X_k,\sum_{j\in\Lambda} X_j^\ast] =\left(\frac1{\sqrt{|\Lambda|}}\right)^2 \sum_{k\in\Lambda} id=\left(\frac1{\sqrt{|\Lambda|}}\right)^2 |\Lambda|
 = id .\]
    Consider the following quadratic Hamiltonian
    \[U_\Lambda=  \mathbf{X}_\Lambda^\ast \mathbf{X}_\Lambda=\frac1{|\Lambda|}\sum_{j,k\in\Lambda} X_j^\ast X_k \]
    %=   \left| \frac1{\sqrt{|\Lambda|}}\sum_{ k\in\Lambda}  X_k\right|^2
 %
Then the corresponding modular dynamics of $\mathbf{X}_\Lambda$ and $\mathbf{X}_\Lambda^\ast$ associated  to $U_\Lambda$ is given by
\[ \alpha_{t,\Lambda} (\mathbf{X}_\Lambda)=e^{i\beta t}\mathbf{X}_\Lambda\qquad and \qquad \alpha_{t,\Lambda} (\mathbf{X}_\Lambda^\ast)=e^{-i\beta t}\mathbf{X}_\Lambda^\ast  \]
 and we have the following result.
   
\begin{proposition} \label{Pro 7.1}
The Dirichlet form in the directions of $\mathbf{X}_\Lambda$ and $\mathbf{X}_\Lambda^\ast$ with respect to $U_\Lambda$ is given by
\begin{equation}{\label{DF1}}
    \mathcal{E}_\Lambda(f)= \hat{\eta} (0) \left( \langle \delta_{\mathbf{X}_\Lambda} (f), \delta_{\mathbf{X}_\Lambda} (f)
\rangle_{\omega}+\langle \delta_{ \mathbf{X}_\Lambda^\ast} (f), \delta_{\mathbf{X}_\Lambda^\ast} (f)
\rangle_{\omega} \right). 
\end{equation}
where $\hat{\eta} (0)=\int_{\mathbb{R}} \eta(t)e^{ist}dt$ .
%is the Fourier transform \footnotetext[1]{Fourier transform notation: $\hat{\eta} (s)=\int_{\mathbb{R}} \eta(t)e^{ist}dt$. } of $\eta(t)$ at $t=0$. 
This form has the dense domain $D(\mathcal{E}_\Lambda)\supset \mathcal{D}_\Lambda$. 
%defined in the configuration space, is symmetric and positive definite and hence, closable.
\end{proposition}

Proof: 
We compute the derivations

\[\delta_{\alpha_t \left(\mathbf{X}_\Lambda\right)} (f) = i [\alpha_t \left(\mathbf{X}_\Lambda\right), f] = {ie}^{i \beta  t} [\left(\mathbf{X}_\Lambda\right), f]\]
\[(\delta_{\alpha_t \left(\mathbf{X}_\Lambda\right)} (f))^{\ast} = - i [f^{\ast}, \alpha_t \left(\mathbf{X}_\Lambda\right)^{\ast}]
= - e^{- i \beta t} i [f^{\ast}, \left(\mathbf{X}_\Lambda\right)^{\ast}]\]

Inserting this into the formula for the Dirichlet for we arrive at \eqref{DF1}$\blacksquare$

The associated Markovian generator  is given by
\begin{equation}{\label{L}}
 \begin{split}
-\mathfrak{L}_\Lambda(f)= %\\
\hat{\eta} (0) \left(-e^{-\frac{\beta}{2} }[\mathbf{X}_\Lambda,f] \mathbf{X}_\Lambda^{\ast}+e^{\frac{\beta}{2} } \mathbf{X}_\Lambda^{\ast}[\mathbf{X}_\Lambda,f]-e^{\frac{\beta}{2} }[\mathbf{X}_\Lambda^\ast,f] \mathbf{X}_\Lambda+e^{-\frac{\beta}{2} }\mathbf{X}_\Lambda [\mathbf{X}_\Lambda^\ast,f]  \right)
 \end{split}   
\end{equation}

%The generator that we are obtaining here is similar to the one in quantum OU.
% \begin{theorem}
 %The following limit 
 %\[\mathcal{E}(f)\equiv \lim_{\Lambda\to \mathbb{Z}^d}\mathcal{E}_\Lambda(f)\]
 %exists and is well defined on a dense set of polynomials  in $X_j, X_j^\ast$, $j\in \mathbb{Z}^d$.
 %Thus the closure of this quadratic form is a Dirichlet form.
 %\end{theorem}
 i.e. the corresponding Dirichlet operator is the generator of quantum O-U semigroup which maps symmetric polynomials in creation/anihilation operators into itself.  By \cite{CFL} the Poincar{\'e} inequality holds and by \cite{CaS} the Log Sobolev inequality is satisfied. The limiting theory can be described in the framework of \cite{RW}.
\end{example}

%%%%e.g. MF2

\begin{example}(Mean Field Model 2){\label{ex1Bis}} %Example 7.2
With $[X_j,X_k^\ast]=\delta_{j,k}$, an integer $n>1$ and $\varepsilon\in[0,1]$, define
\[
\mathbf{X}_{n,\Lambda} \equiv \frac1{|\Lambda|^\varepsilon}\sum_{k\in\Lambda} X_k^n
\]
Then we have
 \[\left[\mathbf{X}_{n,\Lambda}, \mathbf{X}_{n,\Lambda}^\ast \right] = \frac1{ |\Lambda|^{2\varepsilon}}   \sum_{k\in\Lambda} [X_k^n,\sum_{j\in\Lambda} X_j^{\ast n}] =\frac1{ |\Lambda|^{2\varepsilon}}  \sum_{k\in\Lambda} [X_k^n,  X_k^{\ast n}]= \frac1{ |\Lambda|^{2\varepsilon}}  \sum_{k\in\Lambda}P_n(N_k);\]
with some  polynomial $P_n$.
\\
%
%Thus to {\color{red}have a stable in the size of the system commutator it is required that $\varepsilon=\frac12$}.\\
    Consider the following quadratic Hamiltonian
    \[U_\Lambda=  \mathbf{X}_\Lambda^\ast \mathbf{X}_\Lambda=\frac1{|\Lambda|}\sum_{j,k\in\Lambda} X_j^\ast X_k \]
    %=   \left| \frac1{\sqrt{|\Lambda|}}\sum_{ k\in\Lambda}  X_k\right|^2
(which commutes with $\sum_{k\in \Lambda}U_k$).
Now we obtain %Check for \varepsilon. Is it assumed \varepsilon=1/2 

\[\begin{split}
[U_\Lambda,\mathbf{X}_{n,\Lambda}]&\equiv -i\delta_{U_\Lambda}(\mathbf{X}_{n,\Lambda})=
%\frac1{|\Lambda|}[\sum_{j,k\in\Lambda} X_j^\ast X_k ,\frac1{\sqrt{|\Lambda|}}\sum_{l\in\Lambda} X_l^n]
\frac1{|\Lambda|^{\frac12+\varepsilon}}\sum_{j \in\Lambda}\left[ X_j^\ast  ,  X_j^n\right]
\frac1{\sqrt{|\Lambda|}} \sum_{ k\in\Lambda}X_k\\
%&=-n \frac1{|\Lambda|^{\frac12+\varepsilon}} \sum_{j \in\Lambda}  X_j^{n-1}\frac1{\sqrt{|\Lambda|}} \sum_{ k\in\Lambda}X_k  \\
&=-n \frac1{|\Lambda|^{\frac12}}  \mathbf{X}_{n-1,\Lambda}
\mathbf{X}_{\Lambda}
\end{split}
\]
and
\[\begin{split}
(-i\delta_{U_\Lambda})^2(\mathbf{X}_{n,\Lambda}) 
&=-i\delta_{U_\Lambda}\left(-n \frac1{|\Lambda|^{\frac12}}  \mathbf{X}_{n-1,\Lambda}
\mathbf{X}_{\Lambda}\right)\\
&= n(n-1) \frac1{|\Lambda| }  \mathbf{X}_{n-2,\Lambda}
\mathbf{X}_{\Lambda}^2 + n \frac1{|\Lambda|^{\frac12}}  \mathbf{X}_{n-1,\Lambda}
\mathbf{X}_{\Lambda}
\end{split}
\]
Using the formula
\[
\delta_{U_\Lambda}(\mathbf{X}_{\Lambda}^m) = -m \mathbf{X}_{\Lambda}^m
\]
we have
\[\begin{split}
(-i\delta_{U_\Lambda})^3(\mathbf{X}_{n,\Lambda}) 
&= (-i\delta_{U_\Lambda})\left(n(n-1) \frac1{|\Lambda| }  \mathbf{X}_{n-2,\Lambda}
\mathbf{X}_{\Lambda}^2 + n \frac1{|\Lambda|^{\frac12}}  \mathbf{X}_{n-1,\Lambda}
\mathbf{X}_{\Lambda} \right)\\
&=\left(-n(n-1)(n-2)\frac1{|\Lambda|^{\frac3{2}} }  \mathbf{X}_{n-3,\Lambda} 
\mathbf{X}_{\Lambda}^3  -3n(n-1)  \frac1{|\Lambda| }  \mathbf{X}_{n-2,\Lambda} 
\mathbf{X}_{\Lambda}^2 - n \frac1{|\Lambda|^{\frac12}}  \mathbf{X}_{n-1,\Lambda}
\mathbf{X}_{\Lambda} \right)
\end{split}
\]
%%%Check this below
By induction , for $k = n-1$, we get 
\[ 
\delta_{U_\Lambda}^{n-1}(\mathbf{X}_{n,\Lambda})
 =\sum_{l=1}^{n-1}(-i)^l %{\color{red}
 a_{k,l} %}
 \frac1{|\Lambda|^{\frac{l}2}}  \mathbf{X}_{n-l,\Lambda}
\mathbf{X}_{\Lambda}^l 
\]
with some coefficients $a_{k,l}$.
For $k>n-1$ we will not produce different operators,
just the coefficient will be changing according to the following formula
\[ 
 a_{k+1,l}=-(l-1) a_{k,l-1} -(n-l) a_{k,l},
\]
and thus asymptotically they will be growing at most as $2^k n^k$. Additionally each term $\mathbf{X}_{n-l,\Lambda}
\mathbf{X}_{\Lambda}^l$ comes accompanied by a factor 
$\frac1{|\Lambda|^{\frac{l}2}}$
\begin{lemma} \label{Lem 7.1}
The modular dynamics of $\mathbf{X}_{n,\Lambda}$ and $\mathbf{X}_{n,\Lambda}^\ast$ associated  to $U_\Lambda$ given by
\[ 
\alpha_{t,\Lambda} (\mathbf{X}_{n,\Lambda}^\sharp)=\sum_{k=0}^\infty \frac{1}{k!}(-it)^k\delta_{U_\Lambda}^{k}(\mathbf{X}_{n,\Lambda}^\sharp)  \]
is well defined. %{\color{blue}How?}.
\end{lemma}
Hence we conclude with the following result.
\begin{theorem}
The Dirichlet form in the directions of $\mathbf{X}_\Lambda$ and $\mathbf{X}_\Lambda^\ast$ with respect to $U_{\Lambda}$ is given by
\begin{equation}{\label{DF1}}
    \mathcal{E}_\Lambda(f)= \int_\mathbb{R} \left(\langle \delta_{\alpha_t(\mathbf{X}_{n,\Lambda})}(f),\delta_{\alpha_t(\mathbf{X}_{n,\Lambda})}(f)\rangle  + \langle \delta_{\alpha_t(\mathbf{X}_{n,\Lambda}^\ast)}(f),\delta_{\alpha_t(\mathbf{X}_{n,\Lambda}^\ast)}(f)\rangle\right)\eta(t)dt. 
\end{equation}
is well defined on the dense domain $D(\mathcal{E}_\Lambda)\supset \mathcal{D}_\Lambda$ and closable, and hence defines a Markov generator.  
\end{theorem}

 \begin{remark}
In the limit $\Lambda\to \mathbb{Z}^d$ this form for $n>1$ is becoming trivial. 
 It is an interesting open question whether or not the limits $t\to \infty$ and $\Lambda\to \mathbb{Z}^d$
 are interchangeable. 
 \end{remark}
\end{example}

\subsection{Non-diagonal Dirichlet Forms}
In this part we consider the dissipative dynamics defined  by a Dirichlet form with nonlocal derivations. In particular for an infinite product space we provide an example in which the generator has no spectral gap. 
\begin{example}[Z-type fields]{\label{ex2}} %Example 7.3
For   $\mathbf{\kappa}\equiv\{\kappa_j \in\mathbb{C}:  \sum_j|\kappa_j|<\infty\}$, define 
\begin{equation}{\label{1}}
    Z_{\mathbf{\kappa}}=\sum_j\kappa_j A_j 
\end{equation}
the series being convergent in any $\mathbb{L}_{p,\omega_0}$  for $p\in[1,\infty)$. Later on we assume that $\mathbf{\kappa}$ is not equal to a zero vector. % {\color{red}$\boldsymbol{\theta}$}.
Given two absolutely convergent sequences $\mathbf{\kappa}$ and $\mathbf{\xi}$, we have the following CCR relation 
\begin{equation} \label{ZCR1}
[Z_{\mathbf{\kappa}},Z_{\mathbf{\xi}}^\ast]=\sum_j \kappa_j\bar\xi_j.
\end{equation}
The modular dynamics (associated to $\omega_0$) of $Z_{\mathbf{\kappa}}$  is given 
\[ \alpha_t (Z_{\mathbf{\kappa}})=e^{i\beta t} Z_{\mathbf{\kappa}}\:\: and\:\:   \alpha_t (Z_{\mathbf{\kappa}}^\ast)=e^{-i\beta t} Z_{\mathbf{\kappa}}^\ast
\]
%{\color{red}
Given a translation $(T_j\mathbf{\kappa})\equiv (\mathbf{\kappa}_{l-j})_{l\in {\mathbb{Z}}^d}$, we define translation of $Z_{\mathbf{\kappa}}$ by 
$ Z_{T_j\mathbf{\kappa}}$. 
%} 

\begin{theorem} \label{Thm 7.2}
Suppose $\mathbf{\kappa},\mathbf{\xi}\neq \boldsymbol{\theta}$.
The Dirichlet form associated to the directions of $Z_{T_j\mathbf{\kappa}}$ and $Z_{T_j\mathbf{\xi}}^\ast$, $j\in\mathbb{Z}^d$, with respect to the state $\omega_0$ is given by 
\[\mathcal{E} (f)= \hat{\eta} (0)\sum_{ j \in\mathbb{Z}^d} \left( \langle \delta_{Z_{ T_j\mathbf{\kappa}}} (f), \delta_{Z_{T_j\mathbf{\kappa}}} (f)
\rangle_{\omega}+\langle \delta_{ Z_{T_j\mathbf{\xi}}^\ast} (f), \delta_{Z_{T_j\mathbf{\xi}}^\ast} (f)
\rangle_{\omega} \right).\]
with a dense domain
$D(\mathcal{E})\supset \mathcal{D}_\varepsilon$
is closable and hence defines a Markov generator.\\
\bigskip

\noindent {Remark}: Similar conclusion holds for the Dirichlet forms associated to the Gibbs states
$\omega$.  
\end{theorem}

 \begin{proof} 
Thanks to summability of $\mathbf{\kappa}$ and $\mathbf{\xi}$ our Dirichlet form can be dominated by a multiple of the one defined with derivations in directions 
$A_j$ and $A_j^\ast$. For the latter Dirichlet form we already showed
that the domain contains a dense set $\mathcal{D}_\varepsilon$ on which the adjoint operators are well defined and so one can define a %{\color{red}
pre-Markov generator %} 
on the same dense domain. The Friedrichs extension completes the construction of the Markov generator. In case of a Gibbs state different than a product state the Dirichlet form generally the integral with $\eta$ does not factorises, however, as we discussed earlier in section 2, all the arguments go through thanks to finite speed of propagation of information and the fact that $\mathcal{D}_\varepsilon$ is a subset of $\mathbb{L}_{p,\omega}$ for $p\in[1,\infty)$.  

%This Dirichlet form is associated to the self adjoint operator, by Friedrich's extension of operator (positive, densely defined, symmetric operator has self adjoint entension) gives a closed operator.
%
\end{proof}

\end{example}

\begin{example} %Example 7.4
Here we consider an example involving infinite set of CCRs which are not independent in general. 
Given $j\sim k, j,k\in \mathbb{Z}^d$ and $\kappa_j, \varepsilon_k\in\mathbb{C}$ define
\[
Z_{j,k}:=\kappa_j A_j+\varepsilon_k A_k
\]
and consider the following   Hamiltonian
     \[ H_\Lambda=\sum_{j\sim k\atop j,k\in\Lambda} Z_{j,k}^\ast Z_{j,k}=\sum_{j\sim k \atop j,k\in\Lambda}(\bar{\kappa_j} A_j^\ast+\bar{\varepsilon_k}A_k^\ast) (\kappa_j A_j+\varepsilon_k A_k).
     %=\sum_{j\sim katop j,k\in\Lambda}|\kappa_j|^2A_j^\ast A_j+\bar{\kappa_j}\varepsilon_k A_j^\ast A_k+\bar{\varepsilon_k}\kappa_j A_k^\ast A_j+|\varepsilon_k|^2A_k^\ast A_k
     \]
We note that the following derivation 
\[\delta_{\mathbf{H}}(f)\equiv \lim_{\Lambda\to\mathbb{Z}^d}\delta_{H_\Lambda}(f).
\]
 is evidently well defined on all local polynomials in creators and anihilators $A_j^\sharp$, $j\in\mathbb{Z}^d$.\\
 In particular,
 for a fixed $l,m\in\Lambda$,  we have 
\[i{\mathbf{H}}(Z_{l,m})  =
[Z_{l,m},\sum_{j\sim k}Z_{j,k}^\ast Z_{j,k}]= \sum_{j\sim k}[Z_{l,m}, Z_{j,k}^\ast Z_{j,k}]=\sum_{j\sim k}[Z_{l,m}, Z_{j,k}^\ast ]Z_{j,k}\]
\[= \sum_{j\sim k}(\delta_{l,j}|\kappa_j|^2 +\delta_{l,k}\kappa_k\bar{\varepsilon_k}+\delta_{m,j}\varepsilon_j\bar{\kappa_j} +\delta_{m,k} |\varepsilon_k|^2)Z_{j,k}\]
Thus, for later application we note that
 the infinite dimensional hamiltonian evolution
$\alpha_t (Z_{l,m}) \equiv \alpha_{t, \mathbf{H}}(Z_{l,m})\equiv e^{- { t\beta  \delta_{\mathbf{H}}} } {Z_{l,m} }$ is well defined and satisfies the following relation
\[
\frac{d}{{dt}} \alpha_t (Z_{l,m}) 
%= \frac{d}{{dt}} e^{- {it} \beta H} (Z_{l,m}) e^{{it} \beta H} = i \beta e^{- {it} \beta H} (Z_{l,m} H)e^{{it} \beta H} - i \beta e^{- {it} \beta H} ( H Z_{l,m})e^{{it} \beta H} 
%
%= i \beta e^{- {it} \beta H} [Z_{l,m}, H] e^{{it}
%\beta H} \]
%\[
%=i \beta e^{- {it} \beta H} (\sum_{j\sim k}(\delta_{l,j}|\kappa_j|^2 +\delta_{l,k}\kappa_k\bar{\varepsilon_k}+\delta_{m,j}\varepsilon_j\bar{\kappa_j} +\delta_{m,k} |\varepsilon_k|^2)Z_{j,k}) e^{{it}
%\beta H}\]
%\[
=i \beta  (\sum_{j\sim k}(\delta_{l,j}
|\kappa_j|^2 +\delta_{l,k}\kappa_k\bar{\varepsilon_k}+\delta_{m,j}\varepsilon_j\bar{\kappa_j} +\delta_{m,k} |\varepsilon_k|^2)\alpha_t(Z_{j,k}) \]

%Let $e^{- {it} \beta H} (Z_{l,m}) e^{{it} \beta H} =
%P_{l,m}$

%The system of equations for $l,m\in\Lambda$ would be
%\[\frac{d}{{dt}} P_{l,m} = i \beta  (\sum_{j\sim k}(\delta_{l,j}|\kappa_j|^2 +\delta_{l,k}\kappa_k\bar{\varepsilon_k}+\delta_{m,j}\varepsilon_j\bar{\kappa_j} +\delta_{m,k} |\varepsilon_k|^2)P_{j,k} \]
%%%%%%%%%
$\,$\\

\noindent \textbf{Remark}: 
To check that corresponding hamiltonian dynamics 
\[\alpha_t (f) = e^{- { t  \delta_{\mathbf{H}}} } f = \lim_{\Lambda\to\mathbb{Z}^d} e^{- { t  \delta_{\mathbf{H}_\Lambda}} } f  =\lim_{\Lambda\to\mathbb{Z}^d} e^{-itH_\Lambda}fe^{itH_\Lambda} \]
is well defined on local polynomials in creators and annihilators,
note that we have
\[\begin{split}
\frac{d}{dt} \alpha_{t,\Lambda} (A_l) &=
-i\alpha_{t,\Lambda} ([H_\Lambda,A_l])
    \\
  &=   -i \sum_{m\in\Lambda\atop m\sim l}\left(   (|\kappa_l|^2+|\varepsilon_l|^2) \alpha_{t,\Lambda}(A_l)+(\bar{\kappa_l}\varepsilon_m  +\bar{\varepsilon_l}\kappa_m) \alpha_{t,\Lambda}(A_m)  
  \right).
\end{split}\]
From that should be clear (via arguments given in the previous sections) that the infinite dimensional limit can be performed and that we have finite speed of propagation of interaction in the system.\\
To elucidate this a bit, note that the above equation can be solved in the algebra if %{\color{red}
mollified %}
by dividing by a power of $(1+\varepsilon N_\Lambda)^{-1}$ with $N_\Lambda\equiv \sum_{j\in\Lambda}N_j$.
This is because 
\[ \delta_{H_\Lambda}(N_\Lambda)=0.\]
Consequently for mollified problem one can use iteration scheme in operator norm, otherwise one needs to study convergence in $L_p(\omega_0)$ spaces, $p\in(1,\infty)$.
\vspace{0.55cm}

\textbf{Remark} Note that one has a generalisation of the last relation directly to the infinite dimensions in the sense of commutation of derivations as follows
\[
[\delta_{\mathbf{H}},\delta_{\mathbf{N}}] = 0
\]
where \[\delta_{\mathbf{N}}(f)\equiv \lim_{\Lambda\to\mathbb{Z}^d}\delta_{\mathbf{N}_\Lambda}(f)\]
with the limit on the right hand side on local polynomials in creators and annihilators (weak or strong in a Hilbert space) or in operator norm on the algebra generated by the mollified local polynomials (with local  mollification provided by $N_\mathcal{O}$ with $\mathcal{O}\subset\subset\mathbb{Z}^d$). 
\hfill{$\circ$}\\

With the well defined hamiltonian dynamics for which finite speed of propagation of information property holds we have the following result.
\begin{theorem} \label{Prop_3.2}
    The following Dirichlet form is well defined on the set of local polynomials 
    \[
    \mathcal{E}(f)\equiv \sum_{j\sim k\atop j,k\in\mathbb{Z}^d}\int \left(\langle \delta_{\alpha_{t,\mathbf{H}}(Z_{jk})}(f), \delta_{\alpha_{t,\mathbf{H}}(Z_{jk})}(f)\rangle_{\omega_0} + \langle \delta_{\alpha_{t,\mathbf{H}}(Z_{jk}^\ast)}(f), \delta_{\alpha_{t,\mathbf{H}}(Z_{jk}^\ast)}(f)\rangle_{\omega_0}\right) \eta(t)dt
    \]
    and its closure defines a Markov generator in $\mathbb{L}_2(\omega_0)$.
\end{theorem}
\end{example}
\vspace{0.55cm}

A variation in the similar direction is provided in the following example where we use the mixed representation of CCRs.
\begin{example}[Y-type fields]{\label{ex3}} %Example 7.5
    For some $\mathbf{\kappa}$, $\mathbf{\xi}\neq \boldsymbol{\theta}$,  
    \[Y_{\mathbf{\kappa},\mathbf{\xi}}=Z_\mathbf{\kappa}-Z_\mathbf{\xi}^\ast.\]
    Then
    \[[Y_{\mathbf{\kappa},\mathbf{\xi}},Y_{\mathbf{\kappa},\mathbf{\xi}}^\ast]=\left(|\mathbf{\kappa}|_2^2-|\xi|_2^2\right)\mathbf{id} 
 %{\color{red}\kappa_i,\xi_i,\:sums?}   
 \]
 where $|\mathbf{\kappa}|_2^2\equiv \sum_j |\kappa_j|^2$ and similarly for $|\mathbf{\xi}|_2^2$.
 We also note that
    \[
    [Y_{\mathbf{\kappa},\mathbf{\xi}},  N_l]=  \kappa_l A_l-\bar\xi_l A_l^\ast
    \]
Using the above, we have the following fact.    
\begin{lemma}
The modular dynamics of $Y_{\mathbf{\kappa},\mathbf{\xi}}$ and $Y_{\mathbf{\kappa},\mathbf{\xi}}^\ast$ with respect to the infinite product state is given 
\[ \alpha_t (Y_{\mathbf{\kappa},\mathbf{\xi}}) = e^{i\beta t}Z_{\mathbf{\kappa}} -   e^{-i\beta t}Z_{\mathbf{\xi}}^\ast 
\]
\end{lemma}

\noindent Consider the following corresponding Dirichlet form
\[ \mathcal{E}(f)=\sum_{j\in\mathbb{Z}^d }   \int \langle \delta_{ \alpha_t Y_{T_j\mathbf{\kappa},T_j\mathbf{\xi}}} (f), \delta_{\alpha_t Y_{T_j\mathbf{\kappa},T_j\mathbf{\xi}}} (f)\rangle  {\eta}(t)dt
\]
We have the following representation of this Dirichlet form.
\begin{theorem}
The Dirichlet form in the directions of $Y_{T_j\mathbf{\kappa},T_j\mathbf{\xi}}$ and $Y_{T_j\mathbf{\kappa},T_j\mathbf{\xi}}^\ast$, $j\in\mathbb{Z}^d$, associated to the product state $\omega_0$  is well defined on a dense set including $\mathcal{D}_\varepsilon$ and  is given by
\[\begin{split}
\mathcal{E}(f)=&\sum_{j\in\mathbb{Z}^d  }  \hat{\eta}(0)  \langle \delta_{ Z_{T_j\mathbf{\kappa} }} (f), \delta_{Z_{T_j\mathbf{\kappa}}} (f)\rangle_{\omega_0}
+\sum_{j\in\mathbb{Z}^d   }  \hat{\eta}(0)  \langle \delta_{ Z_{T_j\mathbf{\xi}}} (f), \delta_{Z_{T_j\mathbf{\xi}}} (f)\rangle_{\omega_0} \\
&+\sum_{j\in\mathbb{Z}^d  }  \sqrt{2 \pi}\hat{\eta}(-2\beta)  \langle \delta_{ Z_{T_j\mathbf{\kappa}}} (f), \delta_{Z^\ast_{T_j\mathbf{\xi}}} (f)\rangle_{\omega_0} %\\
%%
%&
+\sum_{j\in\mathbb{Z}^d   }  \sqrt{2 \pi}\hat{\eta}(2\beta)  \langle \delta_{ Z^\ast_{T_j\mathbf{\xi}}} (f), \delta_{Z_{T_j\mathbf{\kappa}}} (f)\rangle_{\omega_0} %\\
%%
%&
\end{split}
\] 
On the domain $\mathcal{D}_\varepsilon$ it defines a pre-Markov generator. 
\end{theorem}
\begin{comment}
    For the specific example of  $\eta(t)=\frac{e^{ibt}}{cosh(8\pi t)}$,
    $  b\in \mathbb{R}$, we have 
    $\hat{\eta}(s)=[8cosh((s+b)/16)]^{-1}$. 
Then  
\[\begin{split}
 \hat{\eta}(0)=[8cosh(b/16)]^{-1} , \quad
  \hat{\eta}(2\beta)=[8cosh((\pm 2\beta+b)/16)]^{-1} . 
    \end{split}
    \]
    \hfill$\circ$
    \\
    Playing with $8$ in the cosh we can define number of eigenvectors along the idea of \cite{CZ} and hence have a possibility of introducing more complicated Dirichlet forms of type ($\mathcal{E}'$).
\end{comment}
 \end{example}

\begin{example} %Example 7.6\\
Let 
\[ W_{j,k} \equiv W_{j,k}^{(n,m)} \equiv A_j^{\ast n} A_k^m . \]
Then, using modular dynamics corresponding to $\mathbf{N}_\Lambda\equiv \sum_{l\in\Lambda}N_l$ with $j,k\in \Lambda$, we have
\[\alpha_t (W_{j,k})=\lim_{\Lambda\to\mathbb{Z}^d}e^{-i\beta t\mathbf{N}_\Lambda}W_{j,k}e^{i\beta t \mathbf{N}_\Lambda}=\alpha_t (A_j^{\ast n})\alpha_t ( A_k^m ) = e^{i(m-n) \beta t} W_{j,k}\]
  
\begin{theorem}
The Dirichlet form  with derivations generated by directions $W_{j,k}^{(n,m)}$   given as follows
 \[
 \mathcal{E}(f)= \hat\eta(0)\sum_{j,k\in \mathbb
 {Z}^d\atop j\sim k} \left(\|\delta_{W_{j,k}}(f)\|_{2,\omega}^2 +\|\delta_{W_{j,k}^\ast}(f)\|_{2,\omega}^2\right)
 \]
 is well defined on the algebra of local polynomials in creators and annihilators.
 Its closure defines a Markov generator.
\end{theorem}
 
%  \\
%\textbf{Remark} 
\noindent Note that for selfadjoint operators
\[ W_{j,k} \equiv A_j^{\ast n} A_k^m+A_k^{\ast m} A_j^n= W_{j,k}^\ast.\]
and  the same modular dynamics as before
 we have 
 \[\alpha_t (W_{j,k})=e^{i(m-n)\beta t}A_j^{\ast n} A_k^m+e^{-i(m-n) \beta t}A_k^{\ast m} A_j^n \]
and hence a more complicated
 Dirichlet form  given as follows  

 \[\mathcal{E}(f)= \sum_{j,k}2(\hat{\eta}(0) \langle \delta_{A_j^{\ast n} A_k^m} (f), \delta_{A_j^{\ast n} A_k^m} (f)
\rangle_{\omega} + \hat{\eta}(2(n-m)\beta ) \langle \delta_{A_j^{\ast n} A_k^m} (f), \delta_{A_k^{\ast m} A_j^n} (f)
\rangle_{\omega}\]
\[+\hat{\eta}(2(m-n)\beta ) \langle \delta_{A_k^{\ast m} A_j^n} (f), \delta_{A_j^{\ast n} A_k^m} (f)
\rangle_{\omega}+ \hat{\eta}(0) \langle \delta_{A_k^{\ast m} A_j^n} (f), \delta_{A_k^{\ast m} A_j^n} (f)
\rangle_{\omega})\]  
  In the special case when $m=n$, the time dependence factorises and for each term we get the same multiplier $\hat{\eta}(0)$. Moreover in  the selfadjoint case $ W_{j,k}  = W_{j,k}^\ast$ the  Dirichlet form is not ergodic, i.e. it vanishes on nonzero elements.
  \\
  When $m=1=n$, we have
\[ W_{j,k} \equiv A_j^\ast A_k + A_k^\ast A_j =W_{j,k}^\ast\]
and they are invariant with respect to modular dynamics.
In general such operators with  different indices do not commute,  but in this special case we have the following nice commutation relations
\begin{equation} 
\begin{split}
[W_{j,k},W_{n,m}^\ast]&=[A_j^\ast A_k + A_k^\ast A_j,A_n^\ast A_m + A_m^\ast A_n]\\
&=\delta_{jm}A_n^\ast A_k +\delta_{kn}A_j^\ast A_m +\delta_{jn}A_m^\ast A_k +\delta_{km}A_j^\ast A_n \\
&\quad 
+\delta_{km}A_n^\ast A_j +\delta_{jn}A_k^\ast A_m +\delta_{kn}A_m^\ast A_j +\delta_{jm}A_k^\ast A_n  \\
&=%\delta_{jm}A_n^\ast A_k +\delta_{jm}A_k^\ast A_n 
\delta_{jm} W_{k,n} %\\
 %+\delta_{kn}A_j^\ast A_m +\delta_{kn}A_m^\ast A_j 
+\delta_{kn} W_{j,m}%\\
 %+\delta_{jn}A_m^\ast A_k +\delta_{jn}A_k^\ast A_m \\
+\delta_{jn}W_{k,m}
 %+\delta_{km}A_j^\ast A_n  +\delta_{km}A_n^\ast A_j  
+\delta_{km} W_{j,n}%\\
\end{split}
\end{equation}
The Dirichlet form is not ergodic and to change that one need to add terms with derivations related to
\[ W_{j,k} \equiv i(A_j^\ast A_k - A_k^\ast A_j)\]
\end{example}

\begin{example} %Example 7.7
Consider
\[ Z_{j,k} \equiv A_j^n- A_k^m \]
With notation ${A_j^\ast} ^n\equiv ({A_j^\ast})^n$, we have
\[
[Z_{j,k},Z_{j,k}^\ast]=[A_j^n- A_k^m,{A_j^\ast}^n- {A_k^\ast}^m]=[A_j^n,{A_j^\ast}^n]+[A_k^m, {A_k^\ast}^m]
\]

Using formula \eqref{R1} (or lemma 5.1 from \cite{CZ}),

\[[A_j^n,{A_j^\ast}^n]%=A_j^n{A_j^\ast}^n-{A_j^\ast}^nA_j^n
=(N+n)(N+n-1)\cdots(N+2)(N+1)-N(N-1)(N-2)\cdots(N-(n-1))\]
Hence

\[[Z_{j,k},Z_{j,k}^\ast]=(N+n)(N+n-1)\cdots(N+2)(N+1)-N(N-1)(N-2)\cdots(N-(n-1))\]
\[+(N+m)(N+m-1)\cdots(N+2)(N+1)-N(N-1)(N-2)\cdots(N-(m-1))\]
\begin{lemma} %Lemma 7.3 %{\color{red}??}
The modular dynamics in the directions of $Z_{j,k}$ and $Z_{j,k}^\ast$ with respect to the product state of a system of quantum harmonic oscillators is given by
\[ \alpha_t (Z_{j,k})=e^{in\beta t} A_j^n-e^{im\beta t}A_k^m\:\: and\:\:   \alpha_t (Z_{j,k}^\ast)=e^{-in\beta t}{A_j^\ast}^n-e^{-im\beta t}
{A_k^\ast}^m
\]
\end{lemma}
Proof: We know that for $j,k\in\Lambda$ we have
\[[Z_{j,k},  \sum_{j\in\Lambda}N_j]=[ A_j^n- A_k^m,  \sum_{j\in\Lambda}N_j]=[ A_j^n,  \sum_{j\in\Lambda}N_j]-[ A_k^m,  \sum_{j\in\Lambda}N_j]\]
\[=nA_j^n-mA_k^m\]

And
\[[Z_{j,k}^\ast,  \sum_{j\in\Lambda}N_j]=-n{A_j^\ast}^n+m{A_k^\ast}^m
\]

Then $ \alpha_t (Z_{j,k})=\alpha_t (A_j^n)-\alpha_t (A_k^m)=e^{in\beta t} A_j^n-e^{im\beta t}A_k^m$ and similarly  $ \alpha_t (Z_{j,k}^\ast)=\alpha_t ({A_j^\ast}^n)-\alpha_t ({A_k^\ast}^m)=e^{-in\beta t} {A_j^\ast}^n-e^{-im\beta t}{A_k^\ast}^m$.

\begin{theorem}%Theorem 7.6 % %{\color{red}??}
The Dirichlet form in the directions of $Z_{j,k}$ and $Z_{j,k}^\ast$ with respect to the product state of a system of quantum harmonic oscillators given by 
\[\mathcal{E}_\Lambda(f)=\sum_{j\sim k}\hat{\eta}(0)\langle \delta_{A_j^n} (f), \delta_{A_j^n} (f)
\rangle_{\omega}+\hat{\eta}(0)\langle \delta_{A_k^m} (f), \delta_{A_k^m} (f)
\rangle_{\omega}+\hat{\eta}((m-n)\beta)\langle \delta_{A_j^n} (f), \delta_{A_k^m} (f)
\rangle_{\omega}+\hat{\eta}((n-m)\beta)\langle \delta_{A_k^m} (f), \delta_{A_j^n} (f)
\rangle_{\omega}+\]
\[\hat{\eta}(0)\langle \delta_{A_j^{\ast n}} (f), \delta_{A_j^{\ast n}} (f)
\rangle_{\omega}+\hat{\eta}(0)\langle \delta_{A_k^{\ast m}} (f), \delta_{A_k^{\ast m}} (f)
\rangle_{\omega}-\hat{\eta}((n-m)\beta)\langle \delta_{A_j^{\ast n}} (f), \delta_{A_k^{\ast m}} (f)
\rangle_{\omega}-\hat{\eta}((m-n)\beta)\langle \delta_{A_k^{\ast m}} (f), \delta_{A_j^{\ast n}} (f)
\rangle_{\omega}\]

with a dense domain
$D(\mathcal{E}_\Lambda)\supset \mathcal{D}$
is closable and hence defines a Markov generator.
%defined in the configuration space, is symmetric and positive definite and hence, closable.
    
\end{theorem}
Proof: We compute the derivations in the direction of $\alpha_t (Z_{j,k})$

\[\delta_{\alpha_t (Z_{j,k})} (f) = i [\alpha_t (Z_{j,k}), f] = i e^{in\beta t}[A_j^n, f]- i e^{im\beta t}[A_k^m, f]\]
\[(\delta_{\alpha_t (Z_{j,k})} (f))^{\ast} = - i [f^{\ast}, \alpha_t (Z_{j,k})^{\ast}]
=- ie^{-in\beta t} [f^{\ast},A_j^{\ast n}]+ i e^{-im\beta t}[f^{\ast},A_k^{\ast m}]\]

And hence

\[\int\langle \delta_{\alpha_t (Z_{j,k})} (f), \delta_{\alpha_t (Z_{j,k})} (f)
\rangle_{\omega}\eta(t)dt = \frac{1}{Z} {Tr} (e^{- \beta \sum_{j\in\Lambda}N_j / 2} (\delta_{\alpha_t
(X_j)} (f))^{\ast} e^{-\beta \sum_{j\in\Lambda}N_j / 2} \delta_{\alpha_t (X_j)} (f))\]
\[=\frac{1}{Z} {Tr} (- e^{- \beta \sum_{j\in\Lambda}N_j / 2} (- ie^{-in\beta t} [f^{\ast},A_j^{\ast n}]+ i e^{-im\beta t}[f^{\ast},A_k^{\ast m}])
e^{-\beta \sum_{j\in\Lambda}N_j / 2} ( i e^{in\beta t}[A_j^n, f]- i e^{im\beta t}[A_k^m, f]))\]
\[=\hat{\eta}(0)\langle \delta_{A_j^n} (f), \delta_{A_j^n} (f)
\rangle_{\omega}+\hat{\eta}(0)\langle \delta_{A_k^m} (f), \delta_{A_k^m} (f)
\rangle_{\omega}+\hat{\eta}((m-n)\beta)\langle \delta_{A_j^n} (f), \delta_{A_k^m} (f)
\rangle_{\omega}+\hat{\eta}((n-m)\beta)\langle \delta_{A_k^m} (f), \delta_{A_j^n} (f)
\rangle_{\omega}\]

Similarly

\[\delta_{\alpha_t (Z_{j,k}^\ast)} (f) = i [\alpha_t (Z_{j,k}^\ast), f] = i e^{-in\beta t}[A_j^{\ast n}, f]- i e^{-im\beta t}[A_k^{\ast m}, f]\]
\[(\delta_{\alpha_t (Z_{j,k}^\ast)} (f))^{\ast} = - i [f^{\ast}, \alpha_t (Z_{j,k}^\ast)^{\ast}]
=- ie^{in\beta t} [f^{\ast},A_j^n]+ i e^{im\beta t}[f^{\ast},A_k^m]\]

And hence

\[\int\langle \delta_{\alpha_t (Z_{j,k}^\ast)} (f), \delta_{\alpha_t (Z_{j,k}^\ast)} (f)
\rangle_{\omega}\eta(t)dt\]

\[=\hat{\eta}(0)\langle \delta_{A_j^{\ast n}} (f), \delta_{A_j^{\ast n}} (f)
\rangle_{\omega}+\hat{\eta}(0)\langle \delta_{A_k^{\ast m}} (f), \delta_{A_k^{\ast m}} (f)
\rangle_{\omega}-\hat{\eta}((n-m)\beta)\langle \delta_{A_j^{\ast n}} (f), \delta_{A_k^{\ast m}} (f)
\rangle_{\omega}-\hat{\eta}((m-n)\beta)\langle \delta_{A_k^{\ast m}} (f), \delta_{A_j^{\ast n}} (f)
\rangle_{\omega}\]

Then

\[\mathcal{E}_\Lambda (f) = \sum_{j\sim k} \int^{\infty}_{- \infty} \left(\langle \delta_{\alpha_t(Z_{j,k}) )} (f), \delta_{\alpha_t(Z_{j,k})} (f)\rangle_{\omega}+\langle \delta_{\alpha_t(Z_{j,k}^\ast) )} (f), \delta_{\alpha_t(Z_{j,k}^\ast)} (f)\rangle_{\omega}\right) \eta(t)dt\]

\[=\hat{\eta}(0)\langle \delta_{A_j^n} (f), \delta_{A_j^n} (f)
\rangle_{\omega}+\hat{\eta}(0)\langle \delta_{A_k^m} (f), \delta_{A_k^m} (f)
\rangle_{\omega}+\hat{\eta}((m-n)\beta)\langle \delta_{A_j^n} (f), \delta_{A_k^m} (f)
\rangle_{\omega}+\hat{\eta}((n-m)\beta)\langle \delta_{A_k^m} (f), \delta_{A_j^n} (f)
\rangle_{\omega}+\]
\[\hat{\eta}(0)\langle \delta_{A_j^{\ast n}} (f), \delta_{A_j^{\ast n}} (f)
\rangle_{\omega}+\hat{\eta}(0)\langle \delta_{A_k^{\ast m}} (f), \delta_{A_k^{\ast m}} (f)
\rangle_{\omega}-\hat{\eta}((n-m)\beta)\langle \delta_{A_j^{\ast n}} (f), \delta_{A_k^{\ast m}} (f)
\rangle_{\omega}-\hat{\eta}((m-n)\beta)\langle \delta_{A_k^{\ast m}} (f), \delta_{A_j^{\ast n}} (f)
\rangle_{\omega}\]

\begin{comment}
\begin{remark}${\color{red}??}\,$\\
    For $n=m$, we get

\[\mathcal{E}_\Lambda(f)=\sum_{j\sim k}\hat{\eta}(0)\langle \delta_{A_j^n} (f), \delta_{A_j^n} (f)
\rangle_{\omega}+\hat{\eta}(0)\langle \delta_{A_k^n} (f), \delta_{A_k^n} (f)
\rangle_{\omega}+\hat{\eta}(0)\langle \delta_{A_j^n} (f), \delta_{A_k^n} (f)
\rangle_{\omega}+\hat{\eta}(0)\langle \delta_{A_k^n} (f), \delta_{A_j^n} (f)
\rangle_{\omega}+\]
\[\hat{\eta}(0)\langle \delta_{A_j^{\ast n}} (f), \delta_{A_j^{\ast n}} (f)
\rangle_{\omega}+\hat{\eta}(0)\langle \delta_{A_k^{\ast n}} (f), \delta_{A_k^{\ast n}} (f)
\rangle_{\omega}-\hat{\eta}(0)\langle \delta_{A_j^{\ast n}} (f), \delta_{A_k^{\ast n}} (f)
\rangle_{\omega}-\hat{\eta}(0)\langle \delta_{A_k^{\ast n}} (f), \delta_{A_j^{\ast n}} (f)
\rangle_{\omega}\]

\[=\sum_{j\sim k}\hat{\eta}(0)(\langle \delta_{A_j^n} (f), \delta_{A_j^n} (f)
\rangle_{\omega}+\langle \delta_{A_k^n} (f), \delta_{A_k^n} (f)
\rangle_{\omega}+2Re(\langle \delta_{A_j^n} (f), \delta_{A_k^n} (f)
\rangle_{\omega})\]
\[+\langle \delta_{A_j^{\ast n}} (f), \delta_{A_j^{\ast n}} (f)
\rangle_{\omega}+\langle \delta_{A_k^{\ast n}} (f), \delta_{A_k^{\ast n}} (f)
\rangle_{\omega}-2Re(\langle \delta_{A_j^{\ast n}} (f), \delta_{A_k^{\ast n}} (f)
\rangle_{\omega}))\]

\end{remark}   
\end{comment}

\end{example}
%Using formula \eqref{R1} (or lemma 5.1 from \cite{CZ}),

\begin{example} %Example 7.8
Consider 
    \[Y_{j,k}= A_j^n- A_k^{\ast m}\]
    then  
    \[[Y_{j,k},Y_{j,k}^\ast]=[A_j^n- A_k^{\ast m},A_j^{\ast n}- A_k^{m}]=[A_j^n,A_j^{\ast n}]+[A_k^{\ast m},A_k^{m}]\]
    \[=nA_j^{n-1}A_j^{\ast n-1} -mA_k^{m-1}A_k^{\ast m-1} \]
    and
    \[[Y_{j,k},\sum_jN_j]=[A_j^n- A_k^{\ast m},\sum_jN_j]=nA_j^n+mA_k^{\ast m}\]

    \begin{lemma} %Lemma 7.4 %{\color{red}??}
The modular dynamics of $Y_{j,k}$ and $Y_{j,k}^\ast$ with respect to the infinite product state $\omega_0$ of a system of quantum harmonic oscillators is given 
\[ \alpha_t (Y_{j,k})=e^{in\beta t}A_j^n-e^{-im\beta t}A_k^{\ast m}\:\: and\:\:   \alpha_t (Y_{j,k}^\ast)=e^{-in\beta t} A_j^{\ast n}- e^{im\beta t}A_k^m\]
\end{lemma}
Proof:
So 
\[\alpha_t (Y_{j,k})=\alpha_t(A_j^n)-\alpha_t(A_k^{\ast m})=e^{in\beta t}A_j^n-e^{-im\beta t}A_k^{\ast m}\]
Similarly
\[\alpha_t (Y_{j,k}^\ast)=e^{-in\beta t} A_j^{\ast n}- e^{im\beta t}A_k^m\]

\begin{theorem}%Lemma 7.7 %{\color{red}??}
The Dirichlet form in the directions of $Y_{j,k}$ and $Y_{j,k}^\ast$ with respect to the product state $\omega_0$ is given 
\[\mathcal{E}_\Lambda(f)=\sum_{j\sim k}\hat{\eta}(0)\langle \delta_{A_j^n} (f), \delta_{A_j^n} (f)
\rangle_{\omega}+\hat{\eta}(0)\langle \delta_{A_k^{\ast m}} (f), \delta_{A_k^{\ast m}} (f)
\rangle_{\omega}-\hat{\eta}((-m-n)\beta)\langle \delta_{A_j^n} (f), \delta_{A_k^{\ast m}} (f)
\rangle_{\omega}-\hat{\eta}((n+m)\beta)\langle \delta_{A_k^{\ast m}} (f), \delta_{A_j^n} (f)
\rangle_{\omega}\]
\[+\hat{\eta}(0)\langle \delta_{A_j^{\ast n}} (f), \delta_{A_j^{\ast n}} (f)
\rangle_{\omega}+\hat{\eta}(0)\langle \delta_{A_k^{ m}} (f), \delta_{A_k^{ m}} (f)
\rangle_{\omega}-\hat{\eta}((n+m)\beta)\langle \delta_{A_j^{\ast n}} (f), \delta_{A_k^{ m}} (f)
\rangle_{\omega}-\hat{\eta}((-m-n)\beta)\langle \delta_{A_k^{ m}} (f), \delta_{A_j^{\ast n}} (f)
\rangle_{\omega}\]

\end{theorem}
Proof: 
We compute the derivations in the direction of $\alpha_t (Y_{j,k})$

\[\delta_{\alpha_t (Y_{j,k})} (f) = i [\alpha_t (Y_{j,k}), f] = i e^{in\beta t}[A_j^n, f]- i e^{-im\beta t}[A_k^{\ast m}, f]\]
\[(\delta_{\alpha_t (Y_{j,k})} (f))^{\ast} = - i [f^{\ast}, \alpha_t (Y_{j,k})^{\ast}]
=- ie^{-in\beta t} [f^{\ast},A_j^{\ast n}]+ i e^{im\beta t}[f^{\ast},A_k^{ m}]\]

And hence

\[\int\langle \delta_{\alpha_t (Y_{j,k})} (f), \delta_{\alpha_t (Y_{j,k})} (f)
\rangle_{\omega}\eta(t)dt = \frac{1}{Z} {Tr} (e^{- \beta \sum_{j\in\Lambda}N_j / 2} (\delta_{\alpha_t
(Y_{j,k})} (f))^{\ast} e^{-\beta \sum_{j\in\Lambda}N_j / 2} \delta_{\alpha_t (Y_{j,k})} (f))\]
\[=\frac{1}{Z} {Tr} (e^{- \beta \sum_{j\in\Lambda}N_j / 2} (- ie^{-in\beta t} [f^{\ast},A_j^{\ast n}]+ i e^{im\beta t}[f^{\ast},A_k^{ m}])
e^{-\beta \sum_{j\in\Lambda}N_j / 2} ( i e^{in\beta t}[A_j^n, f]- i e^{-im\beta t}[A_k^{\ast m}, f]))\]
\[=\hat{\eta}(0)\langle \delta_{A_j^n} (f), \delta_{A_j^n} (f)
\rangle_{\omega}+\hat{\eta}(0)\langle \delta_{A_k^{\ast m}} (f), \delta_{A_k^{\ast m}} (f)
\rangle_{\omega}-\hat{\eta}((-m-n)\beta)\langle \delta_{A_j^n} (f), \delta_{A_k^{\ast m}} (f)
\rangle_{\omega}-\hat{\eta}((n+m)\beta)\langle \delta_{A_k^{\ast m}} (f), \delta_{A_j^n} (f)
\rangle_{\omega}\]

Similarly

\[\delta_{\alpha_t (Y_{j,k}^\ast)} (f) = i [\alpha_t (Y_{j,k}^\ast), f] = i e^{-in\beta t}[A_j^{\ast n}, f]- i e^{im\beta t}[A_k^{ m}, f]\]
\[(\delta_{\alpha_t (Y_{j,k}^\ast)} (f))^{\ast} = - i [f^{\ast}, \alpha_t (Y_{j,k}^\ast)^{\ast}]
=- ie^{in\beta t} [f^{\ast},A_j^n]+ i e^{-im\beta t}[f^{\ast},A_k^{\ast m}]\]

And hence

\[\int\langle \delta_{\alpha_t (Y_{j,k}^\ast)} (f), \delta_{\alpha_t (Y_{j,k}^\ast)} (f)
\rangle_{\omega}\eta(t)dt\]

\[=\hat{\eta}(0)\langle \delta_{A_j^{\ast n}} (f), \delta_{A_j^{\ast n}} (f)
\rangle_{\omega}+\hat{\eta}(0)\langle \delta_{A_k^{ m}} (f), \delta_{A_k^{ m}} (f)
\rangle_{\omega}-\hat{\eta}((n+m)\beta)\langle \delta_{A_j^{\ast n}} (f), \delta_{A_k^{ m}} (f)
\rangle_{\omega}-\hat{\eta}((-m-n)\beta)\langle \delta_{A_k^{ m}} (f), \delta_{A_j^{\ast n}} (f)
\rangle_{\omega}\]

Then

\[\mathcal{E}_\Lambda (f) = \sum_{j\sim k} \int^{\infty}_{- \infty} \left(\langle \delta_{\alpha_t(Y_{j,k}) )} (f), \delta_{\alpha_t(Y_{j,k})} (f)\rangle_{\omega}+\langle \delta_{\alpha_t(Y_{j,k}^\ast) )} (f), \delta_{\alpha_t(Y_{j,k}^\ast)} (f)\rangle_{\omega}\right) \eta(t)dt\]

\[=\sum_{j\sim k}\hat{\eta}(0)\langle \delta_{A_j^n} (f), \delta_{A_j^n} (f)
\rangle_{\omega}+\hat{\eta}(0)\langle \delta_{A_k^{\ast m}} (f), \delta_{A_k^{\ast m}} (f)
\rangle_{\omega}-\hat{\eta}((-m-n)\beta)\langle \delta_{A_j^n} (f), \delta_{A_k^{\ast m}} (f)
\rangle_{\omega}-\hat{\eta}((n+m)\beta)\langle \delta_{A_k^{\ast m}} (f), \delta_{A_j^n} (f)
\rangle_{\omega}\]
\[+\hat{\eta}(0)\langle \delta_{A_j^{\ast n}} (f), \delta_{A_j^{\ast n}} (f)
\rangle_{\omega}+\hat{\eta}(0)\langle \delta_{A_k^{ m}} (f), \delta_{A_k^{ m}} (f)
\rangle_{\omega}-\hat{\eta}((n+m)\beta)\langle \delta_{A_j^{\ast n}} (f), \delta_{A_k^{ m}} (f)
\rangle_{\omega}-\hat{\eta}((-m-n)\beta)\langle \delta_{A_k^{ m}} (f), \delta_{A_j^{\ast n}} (f)
\rangle_{\omega}\]

\end{example}

%{\color{red}
\begin{remark} %Remark 7.2
    In Examples 7.5 and 7.8, the operators  $Y_{j,k}$'s satisfy systems of CCRs but they are not necessarily commuting for different pairs of $j,k\in\Lambda$.
\end{remark}%}

%{\color{red}
\begin{remark} %Remark 7.3
    The Markov generators corresponding to these Dirichlet forms are defined on a dense set and have Friedrich's extension to a self adjoint operator denoted by $-\mathfrak{L}$ , which generates a strongly continuous semigroup $P_t=e^{t\mathfrak{L}}.$
\end{remark}

 \begin{example}[\textbf{G}-Models]%Example 7.9
 $\,$\\
For some $\kappa,\xi\neq 0$,  let $Y_{\kappa,\xi}=\kappa A + \xi A^\ast$ and 
\[\mathcal{N}\equiv Y_{\kappa,\xi}^{\ast }Y_{\kappa,\xi}.\]
Consider
    \[G_{\kappa,\xi}\equiv \frac12Y_{\kappa,\xi}^2 = \frac12\left( {\kappa}A- {\xi}A^\ast\right)^2.\]
    %{\color{blue}j?? definition of Z and Y}
    Then, with $\mathcal{R}\equiv |\kappa|^2-|\xi|^2$, we have
    \[
    %\begin{split}
    [G_{\kappa,\xi},G_{\kappa,\xi}^\ast]
    %
   %&= \frac14Y_{\kappa,\xi}[Y_{\kappa,\xi},Y_{\kappa,\xi}^{\ast2}] +\frac14[Y_{\kappa,\xi},Y_{\kappa,\xi}^{\ast2}]Y_{\kappa,\xi}\\
   %  &=\frac14Y_{\kappa,\xi}Y_{\kappa,\xi}^{\ast}[Y_{\kappa,\xi},Y_{\kappa,\xi}^{\ast}]  + \frac14Y_{\kappa,\xi}[Y_{\kappa,\xi},Y_{\kappa,\xi}^{\ast}]Y_{\kappa,\xi}^{\ast}\\
   % &+\frac14Y_{\kappa,\xi}^{\ast }[Y_{\kappa,\xi},Y_{\kappa,\xi}^{\ast }]Y_{\kappa,\xi} 
   % + \frac14[Y_{\kappa,\xi},Y_{\kappa,\xi}^{\ast}]Y_{\kappa,\xi}^{\ast }Y_{\kappa,\xi} \\
   % &= \frac14 2\mathcal{R} \left(Y_{\kappa,\xi}Y_{\kappa,\xi}^{\ast}  + Y_{\kappa,\xi}^{\ast }Y_{\kappa,\xi}\right)\\
   % &
    %= \frac12\mathcal{R} [ Y_{\kappa,\xi},Y_{\kappa,\xi}^{\ast}]  +  \mathcal{R}  Y_{\kappa,\xi}^{\ast }Y_{\kappa,\xi} %\\
    %&    
    %\end{comment}
   % = \frac12\mathcal{R}^2 + \mathcal{R}  Y_{\kappa,\xi}^{\ast }Y_{\kappa,\xi}
       = \frac12\mathcal{R}^2 + \mathcal{R} \mathcal{N}
    %\left(|\mathbf{\kappa}|_2^2-|\xi|_2^2\right)\mathbf{id}.   
    %\end{split}
    \]
 and
    \[%\begin{split}
       [G_{\kappa,\xi},  \mathcal{N}]=   2\mathcal{R} G_{\kappa,\xi}, \qquad 
       [G_{\kappa,\xi}^\ast,  \mathcal{N}]=   -2\mathcal{R} G_{\kappa,\xi}^\ast.   
    %\end{split}
    \]
    Please note that in this model $\mathcal{R}\equiv |\kappa|^2-|\xi|^2$
    can take on positive as well negative values.\\
    It provides another example of type \cite{ShS}.
    More on the corresponding algebra is discussed in \cite{MsZ1}, \cite{MsZ2}.
\\
To introduce infinite dimensional models we consider a product state $\omega_\theta$ associated to one particle interaction of the form $\mathcal{N}_j\equiv \mathcal{N}_j(\kappa_j,\xi_j)$, $j\in{\mathbb{Z}}^d$, as follows.
\[ \omega_\theta \equiv \lim_{\Lambda\to\mathbb{Z}^d} \frac1{Z_\Lambda}Tr_\Lambda e^{-\sum_{j\in\Lambda}\beta_j \mathcal{N}_j(\kappa_j,\xi_j)} 
\]
where $\beta_j\in(0,\infty)$ and $0<Z_\Lambda<\infty$ is the normalisation constant. Thus the constants $\mathcal{R}_j$ may change with $j\in{\mathbb{Z}^d}$ (and the integer lattice can be replaced by a %{\color{red}
triangulation  of a %the universe 
manifold).%}. 
Then the modular and Hamiltonian automorphisms' is defined by the following (non inner) densely defined derivation 
\[
ad_H(f)\equiv \lim_{\Lambda\to \mathbb{Z}^d} \sum_{j\in\Lambda} [\mathcal{N}_j,f]
\]
This setup includes a large family of eigenvectors of modular operator given as follows: For a finite set $J\subset\subset\mathbb{Z}^d$
\[  W_{J,\mathcal{O}}(\mathbf{\kappa},\mathbf{\xi},\mathbf{n})\equiv \prod_{j\in\mathcal{O}} G_{\kappa_j,\xi_j}^{n(j)}\prod_{k\in J\setminus \mathcal{O}} G_{\kappa_k,\xi_k}^{\ast n(k)}\]
where indices $j,k$ indicate independent (commuting) copies of $G_{\kappa,\xi}$, and $\mathbf{\kappa}\equiv (\kappa_j)_{j\in\mathbb{Z}^d},\mathbf{\xi}\equiv (\xi_j)_{j\in\mathbb{Z}^d}$, $\mathbf{n}\equiv\{ n(l)\in\mathbb{N}, l\in\mathbb{Z}^d\}$. 
    \\
    Each of $W_{J,\mathcal{O}}$ allows to introduce a densely defined Dirichlet form as follows
    \[ %{\color{red}
    \mathcal{E}_{J,\mathcal{O}}(f) \equiv \sum_{j\in\mathbb{Z}^d} \left(\|\delta_{W_{J+j,\mathcal{O}+j}}f\|_2 +\|\delta_{W_{J+j,\mathcal{O}+j}^\ast}f\|_2\right) %}
    \]
    where $J+j,\mathcal{O}+j$ denote a shift of $J,\mathcal{O}$ by a vector $j\in\mathbb{Z}^d$.
    Note that in case when $|n(\mathcal{O})|=|n(J\setminus \mathcal{O})|$, the corresponding $W_{J,\mathcal{O}}$ is invariant with respect to modular operator.\\
    
    One can also consider Dirichlet forms with the following non-inner derivations 
    \[ \delta_{\mathcal{W}(J\mathcal{O})}\equiv \sum_{j\in\mathbb{Z}^d}  \delta_{W_{J+j,\mathcal{O}+j}}. \]
\end{example}
 
\section{No Spectral Gap Property}
We can prove, generalising the results of commutative case obtained in \cite{INZ} , that the quantum dissipative semigroups involving Dirichlet forms $\mathcal{E}$
defined with
$Z_{j,k} = A_j-A_k$ and $Y_{j,k} = A_j-A_k^\ast$, for nearest neighbours' pairs $(j,k)$,  decay to equilibrium only polynomially in time. 
We have the following result.
\begin{theorem}
%For $\kappa_j=1=\varepsilon_j$ , $j\in\mathbb{Z}^d$, t
The Poincar{\'e} inequality does not hold for Dirichlet forms $\mathcal{E}$  defined with $Z_{j,k} = A_j-A_k$ as well as with $Y_{j,k} = A_j-A_k^\ast$. \end{theorem}
We remark that so far the only quantum example where Poincar\'e inequality failed was provided in \cite{CFL} for O-U generator with equal coefficients for both directions of the derivations. For that example however an equilibrium state does not exist, while in our example it does. The idea of the proof of the theorem is to show that no constant $m\in(0,\infty)$ exists such that
    \[
   m \|f-\omega(f)\|_2^2 \leq \mathcal{E}(f).
    \]
This is done by constructing a sequence of operators $f_n$ localised in boxes $\Lambda_n$ of size $2n+1$ with variance growing as the volume of the box and the Dirichlet form only as a surface of the box.\\
%The proof and hence the result generalises to all models in which the state $\omega$ has summable decay of correlations.
 \begin{proof} 
     For $Z_{j,k} =\frac12 (A_j- A_k)$,
consider an increasing  sequence of a finite set $\Lambda_n\equiv [-n,n]^d$ and corresponding sequence of the following operators
\[
F_{n}\equiv \sum_{k\in \Lambda_n} A_k^\ast .
\]
Then we have
\[ 
    \delta_{Z_{j,k}^\ast}(F_n) = 0 %i[\frac12 (A_j^\ast- A_k^\ast), \sum_{l\in [-n,n]} A_l^\ast]=0
    \]
    and
\[ 
    \delta_{Z_{j,k}}(F_n) = i[\frac12 (A_j- A_k), \sum_{l\in [-n,n]^d} A_l^\ast] = \frac{i}{2} \sum_{l\in [-n,n]^d} \left(\delta_{j,l} - \delta_{k,l}\right) \]
    with
    %{\color{red}
\[    \sum_{l\in [-n,n]^d} \left(\delta_{j,l} - \delta_{k,l}\right)= 
% \sim 
       % \pm \frac{i}{2} & \text{if }  j\;  or\,k\; is\; outside \, \Lambda_n\\
        0  \text{ if } both\,j,k\,  are\,  outside\,  or\,  both\,  inside\, \Lambda_n\    
\] %}
and otherwise $\sum_{l\in [-n,n]^d} |\left(\delta_{j,l} - \delta_{k,l}\right)|\leq 2d$.
Hence we get

\[
 \mathcal{E}(F_n)= \hat\eta(0)\sum_{j,k\in \mathbb
 {Z}} \left(\|\delta_{Z_{j,k}}(F_n)\|_{2,\omega}^2 +\|\delta_{Z_{j,k}^\ast}(F_n)\|_{2,\omega}^2\right)
 \leq Const |\partial \Lambda_n| 
 \]

On the other hand, for the product state,
\[
\| F_n-\langle F_n\rangle_\omega \|_{2,\omega}^2\sim \sum_{j \in \Lambda_n} \| A_j^\ast-\langle A_j^\ast \rangle_\omega \|_{2,\omega}^2\sim Vol( \Lambda_n) \| A_0^\ast-\langle A_0^\ast \rangle_\omega \|_{2,\omega}^2
\]
Since the boundary  $|\partial \Lambda_n|/Vol( \Lambda_n)\to_{n\to \infty}0$, the Poincare inequality cannot hold.
 
\end{proof}

\noindent
\textbf{Further examples}\\
Similar conclusion holds for Dirichlet form defined with some other derivations. 
\begin{example}
    For $Z_{j,k} =A_j^m- A_k^m$, one can consider a sequence of the following operators
\[
F_{n}\equiv \sum_{k\in \Lambda_n} A_k^\ast .
\]
for a finite set $\Lambda_n$.\\
Then we have
\[ 
    \delta_{Z_{j,k}}(F_n) = i[A_j^m- A_k^m, \sum_{l\in [-n,n]^d} A_l^\ast] = i m\sum_{l\in [-n,n]^d}\left(\delta_{j,l} A_j^{m-1} - \delta_{k,l} A_k^{m-1}\right) \]
\[    \sim \begin{cases}
        mA_j^{m-1} or - mA_k^{m-1}  & \text{if\, either }  j\;  or\,k\; is\; outside \, \Lambda_n\\
        0 & \text{if } both\,  are\,  outside\, \Lambda_n
    \end{cases}
\]
and 
\[ 
    \delta_{Z_{j,k}^\ast}(F_n)=0\]
Hence we conclude that
\[\mathcal{E}(F_n)\sim |\partial \Lambda_n| , \quad and \quad \|F_n-\langle F_n\rangle\|^2\sim |\Lambda_n|.\]
which precludes Poincare inequality.
Another important possibility of a sequence  of test operators is provided by
\[
F_{n}\equiv \sum_{k\in \Lambda_n} N_k .
\]
In this case we have
\[ 
    \delta_{Z_{j,k}^\sharp}(F_n) = sign(\sharp) i[{A_j^\sharp}^m- {A_k^\sharp}^m, \sum_{l\in [-n,n]} N_l] = sign(\sharp)i m \left(\delta_{j,l} {A_j^\sharp}^{m} - \delta_{k,l} {A_k^\sharp}^{m}\right) \]
where $sign(\sharp)$ equals $\pm 1$, %{\color{red}
(1 for A and -1 for $A^\ast$?), %}
respectively. Thus we reach the same conclusion.
\end{example}

 \begin{example}
    For $Y_{j,k}=A_j+A_k^\ast$, %(or $A_j-A_k^\ast$ ), 
consider a sequence
\[F_n=\sum_{l,l'\in\Lambda_n\atop l\sim l'}( A_l +A_{l'}^\ast )\]
% %
\textbf{Remark}
One could generalise this example to include $Y_{\kappa,\xi}$, with $\kappa$ and $\xi$ being absolutely summable sequences satisfying the following mean zero condition
\[
\sum_{j\in\mathbb{Z}^d}\kappa_j=0, \qquad \sum_{j\in\mathbb{Z}^d}\xi_j=0
\]
\end{example} 

\section{Algebra of Invariant Derivations}
In the setup associated to the state $\omega_0$, consider the operators of the following form
\[
A(I,J)\equiv \prod_{i\in I} A_i \prod_{j\in J} A_j^\ast.
\]
One can show using e.g. \eqref{R1} that in case $|I|=|J|$
such operators are % {\color{blue}
invariant with respect to the modular dynamics. %}.
(Note that for $I=J$, the operator  $A(I,J)$ is a polynomial in $N_j, j\in J$.) Choosing $I\cap J=\emptyset$, we can define a Dirichlet form as follows
\[
\mathcal{E}_{I,J}(f)\equiv \sum_{k\in\mathbb{Z}^d} \|\delta_{A(I+k,J+k)}f\|_{2,\omega_0}^2.
\]
We have the following general result.
\begin{theorem}
    The Dirichlet form $\mathcal{E}_{I,J}$ does not satisfy Poincare inequality.
\end{theorem}
 \begin{proof} 
One can prove this general result choosing a sequence of operators
\[ F_n \equiv \sum_{j\in\Lambda_n} N_j\]
and noticing that if $ I+k,J+k \subset \Lambda_n$, then
\[ \delta_{A(I+k,J+k)}(F_n) =0.\]
Hence 
\[ \mathcal{E}_{I,J}(F_n)\sim |\partial\Lambda_n|\]
while
\[ \|F_n-\omega_0(F_n)\|_{2,\omega_0}^2\sim |\Lambda_n|.\]
and hence the Poincaré inequality with any fixed positive constant would break down for large $\Lambda_n$.  
\end{proof} 

In fact in the present setup we have an infinite dimensional algebra which is invariant with respect to the modular operator. 
It includes operators $A(I,J)$, but also more general ones.  For example
denoting by
\[
A_I^{n(I)}\equiv \prod_{i\in I}A_i^{n(i)}
\]
with multi-index $n(I)\equiv(n(i)\in\mathbb{N})_{i\in I}$, and setting
$|n(I)|\equiv \sum_{i\in I} n(i)$, we can consider operators of the form
\[
P_{I,J}(N_I,N_J)\left(A_I^{n(I)} \left(A_J^{n(J)}\right)^\ast \right)
\]
with  $|n(I)|=|n(J)|$ and a %{\color{red}
polynomial $P_{I,J}(N_I,N_J)$.
%}.
\\
This provides a possibility to define an infinite number of Dirichlet forms which cannot satisfy Poincare inequality.
\\
Besides the inner derivations indicated above, one can consider also exterior derivations and an infinite dimensional Lie algebra of them.
For example the following limit provides %{\color{red%}
non-inner derivations
%} 
well defined on a dense set of local elements in $\mathbb{L}_{2,\omega_0}$
\[
\delta_{I,J}(f)\equiv \lim_{\Lambda\to\mathbb{Z}^d}\sum_{k\in\Lambda} \delta_{A(I+k,J+k)}(f) .
\]
An interesting open problem is to provide a constructive examples 
of such algebra for Gibbs states obtained with a nontrivial interaction.\\

The Dirichlet forms defined with derivations which are non-inner provide a new challenging area of noncommutative analysis.
%%%NashSobolev type inequalities ?

\section{Polynomial Decay to Equilibrium}
In this section we consider a model with derivations in direction of $Z_{j,k}=A_j-A_k$ , for $j\sim k,\; j,k\in \mathbb{Z}^d$  in the space associated to the product state $\omega_0$ describing a system of infinite number of quantum harmonic oscillators.
 We show the following result generalising the commutative one of \cite{INZ}.
\begin{theorem}
    A quantum system described by the family of generators of the form in the above  decays to equilibrium  algebraically in time in the sense that
    \[\sum_j|\delta_{A_j^\sharp} (e^{-t {L}}f)|^2\to_{t\to\infty} 0\]
    with algebraic rate.
\end{theorem}
Proof:\\
\textbf{Z-Case}\\
%\[\sum_j|\delta_{A_j} (e^{tL}f)|^2\to 0\]
 %paper by Inglis, Nekludov and me.

For the model with derivations in direction of $Z_{j,k}=A_j-A_k$ , for $j\sim k,\; j,k\in \mathbb{Z}^d$ , 
set $f_t=P_tf$. 

Same as writing the derivative of the scalar product
\begin{equation}\label{Z1D}
  \frac{d}{dt}\langle\delta_{A_{j}}(P_tf),\delta_{A_{j}}(P_tf)\rangle_\omega=\langle\frac{d}{dt}\delta_{A_{j}}(P_tf),\delta_{A_{j}}(P_tf)\rangle_\omega+\langle\delta_{A_{j}}(P_tf),\frac{d}{dt}\delta_{A_{j}}(P_tf)\rangle_\omega 
\end{equation} 
\[=\langle\delta_{A_{j}} {L}(P_tf),\delta_{A_{j}}(P_tf)\rangle_\omega+\langle\delta_{A_{j}}(P_tf),\delta_{A_{j}}\mathfrak{L}(P_tf)\rangle_\omega\]
\[=\langle {L} \delta_{A_j}(P_{t} f)+[ \delta_{A_j}, {L}](P_{t} f),\delta_{A_{j}}(P_tf)\rangle_\omega+\langle\delta_{A_{j}}(P_tf), {L} \delta_{A_j}(P_{t} f)+[ \delta_{A_j}, {L}](P_{t} f)\rangle_\omega\]
\[\leq \langle[ \delta_{A_j}, {L}](P_{t} f),\delta_{A_{j}}(P_tf)\rangle_\omega+\langle\delta_{A_{j}}(P_tf),[ \delta_{A_j}, {L}](P_{t} f)\rangle_\omega\]
%\[\leq 2Re(\langle\delta_{A_{j}}(P_tf),[ \delta_{A_j}, {L}](P_{t} f)\rangle_\omega).\]

Next we use the following fact.
\begin{lemma}\label{Lm8}
\[[\delta_{A_j},L]=
   4\hat{\eta}(0)\sinh(\beta/2)\sum_{k\sim j}\left( \delta_{A_k}f - \delta_{A_j}f \right).\]  
\end{lemma}
\begin{proof} (\textit{of Lemma})
Since 
\[ L = -\hat{\eta} (0) \sum_{j\sim k} \left(\delta^{\star}_{Z_{j,k}} \delta_{ Z_{j,k}}+\delta^{\star}_{Z_{j,k}^\ast} \delta_{Z_{j,k}^\ast}\right) \]
and  $A_l$ commute with $Z_{jk}$,
for fixed $l$, we have
\[ 
  [\delta_{A_l},L] =-\hat{\eta} (0) \sum_{j\sim k}\left([\delta_{A_l}, \delta^{\star}_{Z_{j,k}} ]\delta_{ Z_{j,k}}  
  +[\delta_{A_l}, \delta^{\star}_{Z_{j,k}^\ast}] \delta_{ Z_{j,k}^\ast}  \right). \]
Since, by Proposition {\ref{AdjointPro}}, 
% and the fact that $Z_{j,k}\in\mathcal{D}_\varepsilon$ 
we have
\[\begin{split}
   \delta^{\star}_{Z_{j,k}} (g) 
   &= - \delta_{\alpha_{-i/2}(Z_{j,k}^{\ast})}(g) +i\left(\alpha_{-i/2}(Z_{j,k}^{\ast})- \alpha_{i/2}(Z_{j,k}^{\ast})  \right)g \\
   %&=  - e^{-\frac\beta{2}} \delta_{Z_{j,k}^{\ast}}(g) +i\left(e^{-\frac\beta{2}} Z_{j,k}^{\ast} - e^{\frac\beta{2}} Z_{j,k}^{\ast}   \right)g \\
   &=  - e^{-\frac\beta{2}} \delta_{Z_{j,k}^{\ast}}(g) -2i \sinh(\frac\beta{2}) Z_{j,k}^{\ast} g 
\end{split} \]
%on a dense domain $\mathcal{D}(\delta^{\star}_{X})\supset\mathcal{D}$.
In our setup we have 
\[[\delta_{A_l},\delta_{Z_{j,k}^{\ast}}]=i \delta_{[A_l , Z_{j,k}^{\ast}]} =0.\]
On the other hand for the left multiplication operator $l_{Z_{j,k}^{\ast}}$ by 
$ Z_{j,k}^{\ast}$, thanks to Leibnitz rule for the derivation 
\[  [\delta_{A_l}, L_{Z_{j,k}^{\ast}}]
 = \delta_{A_l}(Z_{j,k}^{\ast}) =
 \delta_{l,j}-\delta_{l,k}
 \]
 Hence we obtain 
\[\begin{split}
    [\delta_{A_l}, \delta^{\star}_{Z_{j,k}} ](g) =
    [\delta_{A_l}, - e^{-\frac\beta{2}} \delta_{Z_{j,k}^{\ast}}  -2i \sinh(\frac\beta{2}) l_{Z_{j,k}^{\ast}}    ](g)
    %&=\delta_{A_l}( gie^{-\frac{\beta}{2} } Z_{j,k}^{\ast}-ie^{\frac{\beta}{2} } Z_{j,k}^{\ast}g )-( ige^{-\frac{\beta}{2}} Z_{j,k}^{\ast}-ie^{\frac{\beta}{2}} Z_{j,k}^{\ast}g )(\delta_{A_l})\\
   &=  [\delta_{A_l},-2i \sinh(\frac\beta{2}) l_{Z_{j,k}^{\ast}}  ](g)\\
   &=  2  \sinh(\frac\beta{2})(\delta_{l,j}-\delta_{l,k}) g
\end{split}\]
 as
$ [A_l,Z_{j,k}^\ast] %=[ A_l, A_j^\ast- A_k^\ast]=
 =\delta_{l,j}-\delta_{l,k}$.
%
%which in all cases be a constant($0$ or $\pm1$) only, so 
Thus we have
\[[\delta_{A_l}, \delta^{\star}_{Z_{j,k}} ]\delta_{ Z_{j,k}}(f)=2sin h(\beta/2)(\delta_{l,j}-\delta_{l,k} )\delta_{ Z_{j,k}}(f) \]
and hence for a fixed $l$,
\[\begin{split}
[\delta_{A_l},L]&=-2sin h(\beta/2)\hat{\eta}(0)\sum_{j\sim k}(\delta_{l,j}-\delta_{l,k} )\delta_{ Z_{j,k}}(f)\\
&= - 4sin h(\beta/2)\hat{\eta}(0)\sum_{j\sim l} \delta_{ Z_{l,j}}(f)\\
& = - 4sin h(\beta/2)\hat{\eta}(0)\sum_{j\sim l}\left( \delta_{A_l}f - \delta_{A_j}f \right)
\end{split}\]  

\end{proof} %of Lemma 8

Using \eqref{Z1D} and lemma \ref{Lm8}, for any $j\in\mathbb{Z}^d$ we have
\[\frac{d}{dt}\langle\delta_{A_{j}}(P_tf),\delta_{A_{j}}(P_tf)\rangle_\omega
\leq
8\sinh(\beta/2)\hat{\eta}(0)\sum_{k\sim j}\mathfrak{Re} \langle \delta_{A_{j}}(P_tf),(\delta_{A_{k}}-\delta_{A_{j}})(P_{t}f)\rangle .\]
Using Cauchy-Schwartz inequality on the right hand side, this implies the following relation
\[\frac{d}{dt}\|\delta_{A_{j}}(P_tf)\| 
\leq
4\hat{\eta}(0)\sinh(\beta/2)\sum_{k\sim j}  \left(\| \delta_{A_{k}}(P_{t}f)\|  
-\|\delta_{A_{j}} (P_{t}f)\|  \right).
\]
with the norm given by the scalar product. \\
Denoting by $\bigtriangleup$  the lattice Laplacian on $\mathbb{Z}^d$ and setting $F(t,j)=\|\delta_{A_{j}}(P_tf)\|$,  with $C\equiv 4\hat{\eta}(0)\sinh(\beta/2)$, we get the following differential inequality
\[ \frac{d}{dt}F(t,j) \leq C \; (\bigtriangleup F)(t,j)\]
Integrating both sides
\[\int_0^t {d}F(s)\leq \int_0^t C\Delta F(s)ds\]
\[F(t)\leq F(0)+\int_0^t C \Delta F(s)ds,\]
by iteration we arrive at the following bound 
\[F(t)\leq e^{Ct\bigtriangleup}F(0)\]
 from which the algebraic decay follows, see e.g. \cite{INZ}. Similar arguments work and so we get the same algebraic bound for derivation with respect to $A_j^\ast$.\\

 \begin{comment}
 Given the above bound involving  symmetric semigroup, we can apply \cite{INZ}, corollary 6.2, as follows  ....
%\begin{comment}
Alternatively we could

\[\frac{d}{dt}\|\delta_{A_{j}}(P_tf)\|^2=2\delta_{A_j}(P_{t} f)\mathfrak{L} \delta_{A_j}(P_{t} f)+2\delta_{A_j}(P_{t} f)[ \delta_{A_j},\mathfrak{L}](P_{t} f)\]

\[=2\delta_{A_j}(P_{t} f)\mathfrak{L} \delta_{A_j}(P_{t} f)+2\delta_{A_j}(P_{t} f)(2sin h(\beta/2)\sum_{j',k'}(\delta_{j,j'}-\delta_{j,k'} )\delta_{ Z_{j',k'}})(P_{t} f)\]

\[\leq 4sin h(\beta/2)\delta_{A_j}(P_{t} f)(\sum_{j',k'}(\delta_{j,j'}-\delta_{j,k'} )\delta_{ Z_{j',k'}})(P_{t} f)\]
Here $j'$ and $k'$ are nearest neighbours.

We can now write

\[\leq 4sin h(\beta/2)\delta_{A_j}(P_{t} f)(\sum_{j',k'}(\delta_{j,j'}-\delta_{j,k'} )(\delta_{ A_{j'}}-\delta_{ A_{k'}})(P_{t} f)\]
\[= 4sin h(\beta/2)\delta_{A_j}(P_{t} f)(\sum_{j'}(\delta_{ A_{j'}}-\delta_{ A_{j}})(P_{t} f)-\sum_{j'}(\delta_{ A_{j}}-\delta_{ A_{j'}})(P_{t} f)\]
\[= 8sin h(\beta/2)\delta_{A_j}(P_{t} f)(\sum_{j'}(\delta_{ A_{j'}}-\delta_{ A_{j}})(P_{t} f))\]
with $j\sim j'$. % are nearest neighbours.
\%\%\% Koniec Uwagi% 
\end{comment}
 
 Let us note that similarly as in the case of functions, \cite{INZ},
 the space of  linear combinations of $A_j, A_j^\ast$, $j\in\mathbb{Z}^d$,
 is  mapped into itself by the generator $L$ .
 Explicitly we have
\[
\begin{split}
  LA_l &= -\hat{\eta} (0) \sum_{j\sim k} \left( \delta^{\star}_{Z_{j,k}^\ast} \delta_{Z_{j,k}^\ast}\right)A_l \\
  &=  i\hat{\eta} (0) \sum_{j\sim k} \left( \delta^{\star}_{Z_{j,k}^\ast} (\delta_{jl}- \delta_{kl}) \right)
  =   4  \hat{\eta} (0) \sinh(\frac\beta{2}) \sum_{j\sim l} (A_j^\ast - A_l^\ast )\\
  %%%\ast
  LA_l^\ast &= -i\hat{\eta} (0) \sum_{j\sim k} \left(\delta^{\star}_{Z_{j,k}} (\delta_{jl}- \delta_{kl})  \right) = 4 \hat{\eta} (0)\sinh(\frac\beta{2}) \sum_{j\sim l}    (A_j  - A_l  )
\end{split}
\]
where we used 
\[ 
\begin{split}\delta^{\star}_{Z_{j,k}} (g) &= - e^{-\frac\beta{2}} \delta_{Z_{j,k}^{\ast}}(g) -2i \sinh(\frac\beta{2}) Z_{j,k}^{\ast} g \\
\delta^{\star}_{Z_{j,k}^\ast} (g) &= - e^{\frac\beta{2}}\delta_{Z_{j,k} }(g) + 2i \sinh(\frac\beta{2}) Z_{j,k}  g %\\
%&= - \delta_{\alpha_{-i/2}(X^{\ast})}(g) +i\left(\alpha_{-i/2}(X^{\ast})- \alpha_{i/2}(X^{\ast})  \right)g
\end{split}
\]
Hence, with $C\equiv 4  \hat{\eta} (0) \sinh(\frac\beta{2})$, we get
\[\begin{split}
L(A_j\pm A_j^\ast) &= C\sum_{j\sim l} ((A_j\pm A_j^\ast) - (A_l\pm A_l^\ast))\\
&\equiv C(\bigtriangleup (A_\cdot\pm A_\cdot^\ast))_j 
\end{split}
\]
Thus for an operator linear in creators and annihilators 
\[f\equiv \sum_j  \kappa_j(t) (A_j\pm A_j^\ast) \]
we have %{\color{red}??}
\[ \begin{split}
\partial_t f &= C\sum_j  \dot\kappa_j(t) (A_j\pm A_j^\ast)\\
  &=  C\sum_j   \kappa_j(t) (\bigtriangleup (A_\cdot\pm A_\cdot^\ast))_j =  C\sum_j   (\bigtriangleup\kappa)_j(t)  (A_j \pm A_j\cdot^\ast) 
\end{split}
\]
which holds iff
\[ \begin{split}
  \dot\kappa_j(t)  = C\sum_j   (\bigtriangleup\kappa)_j(t)   
\end{split}
\]
i.e.
\[\mathbf{\kappa}(t)=e^{tC\bigtriangleup}\mathbf{\kappa}_0.\]
This provides explicit algebraic decay to equilibrium for operators linear in creators and anihillators.\\

%\[
%Lf= -C \sum_j \left(\kappa_j(t) \sum_{l\sim j} (A_l^\ast - A_j^\ast ) + \xi_j(t) \sum_{j\sim l}  (A_j  - A_l  )
%\right)
%\]
%{\color{red}Tutaj jestem 20230918 this look differently then for functions}....

$Y$-case is similar. \qed

\section{Appendix: $\Gamma_1$}
 
%By definition (*),  we have
%\[ 
%2\Gamma_{1}(f) =  \mathfrak{L} (f^\ast f)-f^\ast\mathfrak{L} ( f)-\mathfrak{L} (f^\ast )f 
%\]
For the generator 
\[\begin{split}
-\mathfrak{L} (f)&=\int   \left(\delta^\star_{\alpha_t(X)}\delta_{\alpha_t(X)}(  f) +   \delta^\star_{\alpha_t(X^\ast)}\delta_{\alpha_t(X^\ast)}(f) \right)\eta(t) dt \\
\end{split}\]
associated to a state $\omega$ and the corresponding  Dirichlet form in this appendix we demonstrate that 
\[ \begin{split}
2\Gamma_{1}(f) &=  \mathfrak{L} (f^\ast f)-f^\ast\mathfrak{L} ( f)-\mathfrak{L} (f^\ast )f \\
&= \int\left(
  \left|\delta_{\alpha_{t-i/4}(X^{\ast})} (f)\right|^2 +\left| \delta_{\alpha_{t-i/4}(X)} (f)\right|^2\right)\left( \eta(t+i/4) +\eta(t-i/4)\right) dt \\
\end{split}
\]
for all operators $f\in\mathcal{D}(\mathfrak{L})$ such that $f^\ast f\in\mathcal{D}(\mathfrak{L})$
and $f^\ast\mathfrak{L}(f)\in\mathbb{L}_2$.
\begin{proof} 
Using the Leibnitz rule for the derivation for the first part of the integrant,
(the second will be analogous), we have
\begin{equation*}\begin{split} \delta^\star_{\alpha_t(X)}\delta_{\alpha_t(X)}(f^\ast f)&=     \delta^\star_{\alpha_t(X)} \left(\delta_{\alpha_t(X)}(f^\ast )f 
+ 
f^\ast \delta_{\alpha_t(X)}(f)\right)  
%\;\,\qquad\qquad{Term 1}\\
%&+
%\delta^\star_{\alpha_t(X^\ast)}\left(\delta_{\alpha_t(X^\ast)}(f^\ast) f  + f^\ast \delta_{\alpha_t(X^\ast)}(f) \right)   \;\;\quad\qquad{Term 2}\\
\end{split} 
\end{equation*}
Since  the 1st, respectively 2nd, formula of Proposition 2 for  modified Leibnitz property  of the adjoint, for $f,g\in\mathcal{D}_\varepsilon$ , we have  
\[ \begin{split}
  \delta^{\star}_{Y} (fg) 
&= \delta^{\star}_{Y} (f) g  - f  \delta_{\alpha_{-i/2}(Y^{\ast})}(g) \\
&= f \delta^{\star}_{Y} (g) - \delta_{\alpha_{i/2}(Y^{\ast})}(f)  g ,
\end{split}
\] 
hence
\[\begin{split} \delta^\star_{\alpha_t(X)}\delta_{\alpha_t(X)}(f^\ast f)&=    \left( \delta^\star_{\alpha_t(X)}  \delta_{\alpha_t(X)}(f^\ast )\right)f 
- 
\delta_{\alpha_t(X)}(f^\ast)\cdot  \delta_{\alpha_{t-i/2}(X^{\ast})} (f)  \\ %\;\,\qquad\qquad{Term\, 1\, A}\\ 
&+  f^\ast \left(\delta^\star_{\alpha_t(X)} \delta_{\alpha_t(X)}(f)\right) 
 - \delta_{\alpha_{t+i/2}(X^{\ast})}(f^\ast)\cdot  \delta_{\alpha_t(X)}(f)   \\
\end{split} \]
Similarly for the term with $X^\ast$ replacing $X$, we get
\[\begin{split} \delta^\star_{\alpha_t(X^\ast)}\delta_{\alpha_t(X^\ast)}(f^\ast f)&=    \left( \delta^\star_{\alpha_t(X^\ast)}  \delta_{\alpha_t(X^\ast)}(f^\ast )\right)f 
- 
\delta_{\alpha_t(X^\ast)}(f^\ast)\cdot  \delta_{\alpha_{t-i/2}(X)} (f)  \\   
&+  f^\ast \left(\delta^\star_{\alpha_t(X^\ast)} \delta_{\alpha_t(X^\ast)}(f)\right) 
 - \delta_{\alpha_{t+i/2}(X)}(f^\ast)\cdot  \delta_{\alpha_t(X^\ast)}(f)   \\
 %\;\,\qquad\qquad{Term\, 2\, B}\\ \\
%%%%%%2nd part %%%%%%%%%
\end{split} \]
Adding together both formulas we have

\begin{equation}\leqnomode
\tag{A.1}
\begin{split}
\mathfrak{L} (f^\ast f)&=\int   \left(\delta^\star_{\alpha_t(X)}\delta_{\alpha_t(X)}(f^\ast f) +   \delta^\star_{\alpha_t(X^\ast)}\delta_{\alpha_t(X^\ast)}(f^\ast f) \right)\eta(t) dt \\
&= f^\ast \mathfrak{L} (f) + \mathfrak{L} (f^\ast) f \\
& + \int\left(
\delta_{\alpha_t(X)}(f^\ast)\cdot  \delta_{\alpha_{t-i/2}(X^{\ast})} (f) \right)\eta(t)dt\\
&+ \int\left( \delta_{\alpha_{t+i/2}(X^{\ast})}(f^\ast)\cdot  \delta_{\alpha_t(X)}(f)  \right)\eta(t)dt \\
&+ \int\left(
\delta_{\alpha_t(X^\ast)}(f^\ast)\cdot  \delta_{\alpha_{t-i/2}(X)} (f) \right)\eta(t)dt\\
&+ \int\left(\delta_{\alpha_{t+i/2}(X)}(f^\ast)\cdot  \delta_{\alpha_t(X^\ast)}(f)\right)\eta(t)dt
\end{split}    
\end{equation}
Shifting the integration variable in first and third term by $-i/4$ and in the second and fourth by $-i/4$, we get
\[\begin{split}
\mathfrak{L} (f^\ast f)&=\int   \left(\delta^\star_{\alpha_t(X)}\delta_{\alpha_t(X)}(f^\ast f) +   \delta^\star_{\alpha_t(X^\ast)}\delta_{\alpha_t(X^\ast)}(f^\ast f) \right)\eta(t) dt \\
&= f^\ast \mathfrak{L} (f) + \mathfrak{L} (f^\ast) f \\
%1
& + \int\left(
\delta_{\alpha_{t+i/4}(X)}(f^\ast)\cdot  \delta_{\alpha_{t-i/4}(X^{\ast})} (f) \right)\eta(t+i/4)dt\\
%2
&+ \int\left( \delta_{\alpha_{t+i/4}(X^{\ast})}(f^\ast)\cdot  \delta_{\alpha_{t-i/4}(X)}(f)  \right)\eta(t-i/4)dt \\
%3
&+ \int\left(
\delta_{\alpha_{t+i/4}(X^\ast)}(f^\ast)\cdot  \delta_{\alpha_{t-i/4}(X)} (f) \right)\eta(t+i/4)dt\\
%4
&+ \int\left(\delta_{\alpha_{t+i/4}(X)}(f^\ast)\cdot  \delta_{\alpha_{t-i/4}(X^\ast)}(f)\right)\eta(t-i/4)dt
\end{split}
\]
Now the terms in the brackets can be written as squares of operators. After adding 
corresponding terms with the same weight, we obtain
\[\begin{split}
\mathfrak{L} (f^\ast f)&=\int   \left(\delta^\star_{\alpha_t(X)}\delta_{\alpha_t(X)}(f^\ast f) +   \delta^\star_{\alpha_t(X^\ast)}\delta_{\alpha_t(X^\ast)}(f^\ast f) \right)\eta(t) dt \\
&= f^\ast \mathfrak{L} (f) + \mathfrak{L} (f^\ast) f \\
%1+3
& + \int\left(
  \left|\delta_{\alpha_{t-i/4}(X^{\ast})} (f)\right|^2 +\left| \delta_{\alpha_{t-i/4}(X)} (f)\right|^2\right)\eta(t+i/4)dt\\
%2+4
&+ \int\left( \left|  \delta_{\alpha_{t-i/4}(X)}(f)\right|^2  +\left|  \delta_{\alpha_{t-i/4}(X^\ast)}(f)\right|^2\right)\eta(t-i/4)dt \\
\end{split}\]
Since the brackets are the same in both integral expressions, finally this can be rearranged as follows
\[\begin{split}
\mathfrak{L} (f^\ast f)&=\int   \left(\delta^\star_{\alpha_t(X)}\delta_{\alpha_t(X)}(f^\ast f) +   \delta^\star_{\alpha_t(X^\ast)}\delta_{\alpha_t(X^\ast)}(f^\ast f) \right)\eta(t) dt \\
&= f^\ast \mathfrak{L} (f) + \mathfrak{L} (f^\ast) f \\
%1+3
& + \int\left(
  \left|\delta_{\alpha_{t-i/4}(X^{\ast})} (f)\right|^2 +\left| \delta_{\alpha_{t-i/4}(X)} (f)\right|^2\right)\left( \eta(t+i/4) +\eta(t-i/4)\right) dt \\
\end{split}\]
Hence we conclude
\[ \leqnomode
\tag{A.2}
\begin{split}
2\Gamma_{1}(f) &=  \mathfrak{L} (f^\ast f)-f^\ast\mathfrak{L} ( f)-\mathfrak{L} (f^\ast )f \\
&= \int\left(
  \left|\delta_{\alpha_{t-i/4}(X^{\ast})} (f)\right|^2 +\left| \delta_{\alpha_{t-i/4}(X)} (f)\right|^2\right)\left( \eta(t+i/4) +\eta(t-i/4)\right) dt \\
\end{split}
\]
%\qed
%Hence $-\mathfrak{L}$ is a Markov generator.
%%%%%%%%
\end{proof}

If $ \alpha_{\pm i/4}(X)=e^{\pm\frac12\xi}$X  %{\color{red}??}, 
%%%%%%%%%%%%%%%%%%%%%%%%%%
\begin{comment}
%(but not necessary with trivial hamiltonian evolution), 
as e.g. given by
\[
X= \int \alpha_t(B) \eta(t)dt
\]
with analytic function in a strip $\Im z\in (-i/2,+i/2)$ with twisted boundary conditions $\eta(t\pm i/2)= e^{\mp \xi}$,
\end{comment}
then the formula (A.2) yields
\begin{comment}
\begin{equation}\leqnomode
\tag{A.2}
\begin{split}
\mathfrak{L} (f^\ast f)&=\int   \left(\delta^\star_{\alpha_t(X)}\delta_{\alpha_t(X)}(f^\ast f) +   \delta^\star_{\alpha_t(X^\ast)}\delta_{\alpha_t(X^\ast)}(f^\ast f) \right)\eta(t) dt \\
&= f^\ast \mathfrak{L} (f) + \mathfrak{L} (f^\ast) f \\
& + 2 \int \left( 
e^\xi| \delta_{\alpha_{t}(X^{\ast})} (f) |^2 
+ e^{-\xi}   |\delta_{\alpha_t(X)}(f) |^2  \right)\eta(t)dt 
\end{split}    
\end{equation}
Thus we get in this case   
\end{comment}

\[ \leqnomode
\tag{A.3} 
 \frac{1}{C}\Gamma_{1}(f) 
=   
  e^{\xi}\left|\delta_{X^{\ast}}(f)\right|^2 + e^{-\xi}\left| \delta_{X}(f)\right|^2  
\]
with a constant $C\equiv \int \left( \eta(t+i/4) +\eta(t-i/4)\right) dt$.
\\
For example this will be the case of modular dynamics associated to $\omega_0$ and 
derivations associated with $A_{I,J}$.\\
\begin{comment}
In case of 
\[\tilde{\mathcal{E}}(f)\equiv 
\kappa\langle \delta_X(f), \delta_X(f)\rangle + \mu \langle \delta_{X^\ast}(f), \delta_{X^\ast}(f)\rangle
\]
defined with an eigenvector $X$ %$X=\int\alpha_t(B)\eta(t)dt$ 
and some positive constants $\kappa,\mu\in(0,\infty)$,
we have generator 
\[ 
-\tilde{\mathfrak{L}} (f^\ast f) =    \kappa\delta^\star_{X}\delta_{X}(f^\ast f) +   \mu \delta^\star_{X^\ast}\delta_{ X^\ast}(f^\ast f).  
 \]
Then similar computations as above yield
\[\leqnomode
\tag{A.4}
\tilde{\Gamma}_1(f) =  \kappa e^{-\xi}   |\delta_{X}(f) |^2  
+   \mu e^\xi| \delta_{X^{\ast}} (f) |^2 
\]
\end{comment}
Generalisation to the case of infinite number of eigenvectors is as follows
%\[
 %X_j=\left(\int\alpha_t(B_{j,l})\eta(t)dt\right)_{l=1,..,K}, 
%\]
$X_j$,$j\in\mathbb{Z}^d$, is as follows %{\color{red}??}
\[ 
-\tilde{\mathfrak{L}} (f^\ast f) =    \sum_{j\in\mathbb{Z}^d} \left(\kappa_j\delta^\star_{X_j}\delta_{X_j}(f^\ast f) +   \mu_j \delta^\star_{X_j^\ast}\delta_{ X_j^\ast}(f^\ast f)\right)  
 \]
 with some constants $\kappa_j, \mu_j\in\mathbb{R}^+$.
Then similar computations as above yield
\[
\tilde{\Gamma}_1(f) =  \sum_{j\in\mathbb{Z}^d} \left(\kappa_j e^{-\xi_j}   |\delta_{X_j}(f) |^2  
+   \mu_j  e^{\xi_j}| \delta_{X_j^{\ast}} (f) |^2 \right)
\]
In the infinite dimensional case on a lattice to secure the dense domain it is necessary to have the finite speed of propagation of information.
%(as even for localised $B_{j,l}$, the eigenvectors are not localised).
\label{EndGamma1}
\\

%B-E 
\textbf{Remark A.1}:\\
Recall that using Leibnitz rule for the derivationas and modified Leibnitz rule for its adjoint in $\mathbb{L}_2$
\[ \begin{split}
  \delta^{\star}_{Y} (fg) 
&= \delta^{\star}_{Y} (f) g  - f  \delta_{\alpha_{-i/2}(Y^{\ast})}(g) ,
\end{split}
\] 
 we have the following calculations
\[\begin{split}
\delta_{X_t}^\star\delta_{X_t}(f^\ast f) &= \delta_{X_t}^\star\left(\delta_{X_t}(f^\ast) f+f^\ast \delta_{X_t}(f)\right)  \\
& =  \left(\delta_{X_t}^\star\delta_{X_t}(f^\ast)\right) f  - f^\ast  \delta_{\alpha_{-i/2}(X_t^{\ast})}(f) 
  +f^\ast \left(\delta_{X_t}^\star\delta_{X_t}(f)\right) 
- \delta_{\alpha_{i/2}(X_t^{\ast})}(f^{\ast})  f  %\\
%&=\delta_{X_t}^\star\left(\delta_{X_t}(f^\ast) \delta_{X_t}^\star(f) -\delta_{X_t}(f^\ast)  +
\delta_{X_t}^\star(f^\ast) \delta_{X_t}(f) 
\end{split}
\] 
which implies
\[\begin{split}
\mathfrak{L}(f^\ast f)  =  \mathfrak{L}(f^\ast)  f  - f^\ast \int\left( \delta_{\alpha_{-i/2}(X_t^{\ast})}(f)  + \delta_{\alpha_{-i/2}(X_t)}(f) \right)\eta_t dt
  +f^\ast  \mathfrak{L}(f) 
-\int \left(\delta_{\alpha_{i/2}(X_t^{\ast})}(f^{\ast}) +\delta_{\alpha_{i/2}(X_t)}(f^{\ast})\right)\eta_t dt  f .
\end{split}
\] 
provided operator $f$ satisfies suitable conditions so that the operations on the right hand side make sense. In particular $f^\ast f\in \mathcal{D}(\mathfrak{L})$, for bounded operators $f\in\mathcal{D}(\mathfrak{L})$,  if  $\exists C\in(0,\infty)$ such that we have
\[\|\int \left(\delta_{\alpha_{i/2}(X_t^{\ast})}(f^{\ast}) +\delta_{\alpha_{i/2}(X_t)}(f^{\ast})\right)\eta_t dt\|^2\leq 
2\int \left(\|\delta_{\alpha_{\pm i/2}(X_t^{\ast})}(f^{\ast})\|_2^2 +\|\delta_{\pm \alpha_{i/2}(X_t)}(f^{\ast})\|_2^2 \right)\eta_t dt\| \leq C \, \mathcal{E}(f)
\]
The last condition can be satisfied if $\eta$ is analytic in a strip $[-i/2,+i/2]$ and $|\eta (t\pm i/2)|\leq C\, \eta(t)$.
In the models considered before a domain of validity of the above relation could contain local polynomials of elements in $\mathcal{D}_\varepsilon$.
For a positivity preserving contraction semigroup in $\mathbb{L}_2$, for any $t>0$ and $f\in  \mathbb{L}_\infty$, we have
$P_t f\in\mathcal{D}(\mathfrak{L})\cap \mathbb{L}_\infty$. Thus for bounded operators, we can
identify $\Gamma_1$ as follows: For $0<s<t$
\[\frac{d}{ds}\left(P_{t-s} |P_s f|^2-|P_s f|^2\right) =
2 P_{t-s} \Gamma_1(P_s f ) \geq 0 .
\]
Using this, we have
\[
P_t f^2-|P_t f|^2 = \int_0^t  \frac{d}{ds}P_{t-s} (P_s f)^2 ds=
\int_0^t ds \frac{d}{ds}P_{t-s} (P_s f)^2 = \int_0^t ds\,  2 P_{t-s} \Gamma_1(P_s f) \geq 0.
\]
Hence, under the above conditions, for Markov semigroup in $\mathbb{L}_2$ the following 
strong positivity conditions holds.
\begin{proposition} \label{Pro 11.1}
For the diffusion semigroup $P_t\equiv e^{t\mathfrak{L}}$ in $\mathbb{L}_2(\omega)$, for any $t>0$ and $f\in \mathbb{L}_2\cap \mathbb{L}_\infty$, the following Schwartz inequality is true.
\[  |P_t f|^2\leq P_t |f|^2.\] 

\end{proposition}

\textbf{Remark A.2}:\\

Next we note that using the formula for  $\Gamma_1$, with suitable function $\gamma(t)\geq 0$, in general with some positive constants $\nu_\sharp$, we have
\[\begin{split}
    \Gamma_1(\alpha_{-i/4}(f))&= \int \sum_{\sharp}\nu_\sharp\left( |\alpha_{-i/4}(\delta_{\alpha_t(X^\sharp)}(f)|^2\right)\gamma(t)dt\\
   & = \int \sum_{\sharp}\nu_\sharp\left(  \rho^{-1/4}\left(\delta_{\alpha_t(X^\sharp)}(f)\right)^\ast \rho^{1/2}\left(\delta_{\alpha_t(X^\sharp)}(f)\right)\rho^{-1/4}\right)\gamma(t)dt 
\end{split}\]
 Hence %{\color{red}??} 
 we have the following property.
 \begin{proposition} \label{Pro 11.12}
 \[
%{\color{red}??} 
\omega \, \Gamma_1(\alpha_{-i/4}(f)) = %Tr \rho^\frac12 \left(\delta_{\alpha_t(X^\sharp)}(f)\right)^\ast \rho^\frac12\left(\delta_{\alpha_t(X^\sharp)}(f)\right) \gamma(t)dt 
  \int \sum_{\sharp} \nu_\sharp\langle \delta_{\alpha_t(X^\sharp)}(f),\delta_{\alpha_t(X^\sharp)}(f)\rangle \gamma(t) dt 
\simeq\mathcal{E}_\gamma (f).
 \]
     
 \end{proposition}
In case when $X^\sharp$ are eigenvectors of modular operator to the power $\pm1/4$, % using the formula (A.4)
with some positive constants $\nu,\mu$, we have
\[
 {\Gamma}_1(\alpha_{-i/4}(f)) =   \nu e^{-\xi/2}   |\delta_{X}(\alpha_{-i/4}(f)) |^2  
+   \mu e^{\xi/2}| \delta_{X^{\ast}} (\alpha_{-i/4}(f)) |^2 = 
\nu    |\alpha_{-i/4}\left(\delta_{X}(f)\right) |^2  
+   \mu |\alpha_{-i/4}\left(\delta_{X^\ast}(f)\right) |^2 
\]
and hence
\[
\omega\tilde{\Gamma}_1(\alpha_{-i/4}(f)) = \nu\langle\delta_{X}(f),\delta_{X}(f)\rangle
+\mu \langle\delta_{X^\ast}(f),\delta_{X^\ast}(f)\rangle \sim\mathcal{E}(f)
\]

%%%%%%%%%%%%%%%%%%%%%%%%%%%%%
\begin{comment}
We also note that
\[\begin{split}
-\mathfrak{L}_\eta (\alpha_{-i/4}(f))&= \sum_\sharp \int \left(  \left(i(R_{\alpha_{-i/2}({X_t^\sharp}^\ast)} - L_{\alpha_{+i/2}({X_t^\sharp}^\ast)} \right) \delta_{\alpha_t({X_t^\sharp})}(\alpha_{-i/4}(f)f)   \right)\eta(t) dt \\
%&= \alpha_{-i/4}\left(\sum_\sharp \int \left( \left(i(R_{\alpha_{-i/2}({X_t^\sharp}^\ast)} - L_{\alpha_{+i/2}({X_t^\sharp}^\ast)} \right) \delta_{ ({X_t^\sharp})}(f) )   \right)\eta(t-i/4) dt \right)
& =\alpha_{-i/4}\left(-\mathfrak{L}_{\eta(\cdot-i/4)}(f)\right)
\end{split}\]

%%%%%%%%%%%%%%%%%%%%%%
and the fact that $\mathfrak{L}\mathbb{I}=0$, we have
\[ \langle \mathbb{I}, \Gamma_1(\alpha_{i/4}(f))\rangle=\frac12\langle \mathbb{I}, \left( \mathfrak{L} (\alpha_{i/4}(f)^\ast \alpha_{i/4}(f))-\alpha_{i/4}(f)^\ast\mathfrak{L} (\alpha_{i/4}(f))-\mathfrak{L} (\alpha_{i/4}(f)^\ast )\alpha_{i/4}(f)  \right) 
\rangle=\mathcal{E}(f)\]
provided $\Gamma_1(\alpha_{i/4}(f))\in\mathbb{L}_1(\omega)$ which is the case if $\alpha_{i/4}(f),\alpha_{i/4}(f)^\ast \alpha_{i/4}(f)\in\mathcal{D}(\mathfrak{L})$.    
\end{comment}

Hence, we have the following property.
\begin{proposition}
For $f_t\equiv P_tf\equiv e^{-t\mathfrak{L}}f$, if
\[  
\Gamma_{1}(f_t) \leq e^{-2mt}Const\, \Gamma_{1}(f), 
\]
then we have %{\color{red}??}
\[\mathcal{E}(f_t)\leq e^{-2mt} Const\, \mathcal{E}(f)\]
and if Poincar\'e inequality holds,  we get
\[
\|f_t-\omega (f)\|_2^2 \leq e^{-2mt} \|f -\omega (f)\|_2^2.
\]
\end{proposition}

%\subsection{References}

%\hyperlink{thesentence}{any sentence}
\hfill\hyperlink{Contents}{${\color{blue}\bullet}${Back to Contents}}


\begin{thebibliography}{99}

\bibitem[AHK]{AHK} S. Albeverio, R. Hoegh-Krohn,
Dirichlet Forms and Markov Semigroups on C* Algebras, Commun. Math. Phys. 56 (1977) 173-187.
%\url{DOI:10.1007/BF01611502}

\bibitem[BD]{BD} A. Beurling and J. Deny. Dirichlet spaces. Proc. Natl. Acad. Sci., 45(2):208–215, 1959.

\bibitem[BKP]{BKP} C. Bahn, C.K. Ko, Y.M. Park, \newblock{Dirichlet forms and symmetric Markovian semigroups on CCR Algebras with quasi-free states}
\newblock{\it J. Math. Phys. } {\bf 44}, {\rm (2003)}, 723--753. 

\bibitem[BR]{BR}  
O.Bratteli and  D. W. Robinson,  Operator Algebras and Quantum Statistical Mechanics,
Springer-Verlag, New York, Heidelberg, Berlin, vol. I (1979), vol. II (1981).

\bibitem[CaS]{CaS} R. Carbone \&  E. Sasso, Hypercontractivity for a quantum Ornstein–Uhlenbeck semigroup, \newblock{\it Probab. Theory Relat. Fields} {\bf 140(3)}, {\rm  2008},  505-522, \url{DOI:10.1007/s00440-007-0073-2}

\bibitem[CB]{CB} R. Carbone and A. Martinelli, Logarithmic Sobolev inequalities in non-commutative algebras. Infin. Dimens. Anal. Quantum Probab. Relat. Top., Vol. 18, No. 02, 1550011 (2015), \url{https://doi.org/10.1142/S0219025715500113}

\bibitem[Cip1]{Cip1} F. Cipriani, \newblock{Dirichlet forms and Markovian semigroups on standard forms of von Neumann algebras},
\newblock{\it J. Funct. Anal.} {\bf 147} {\rm (1997)}, 259-300.

\bibitem[Cip3]{Cip3} F. Cipriani, \newblock{``Dirichlet forms on noncommutative spaces''},
\newblock{Lecture Notes in Math.} {\bf 1954} {\rm (2008)}, 161-276.

\bibitem[CFL]{CFL} F. Cipriani, F. Fagnola, J.M. Lindsay, \newblock{Spectral analysis and Feller property for quantum Ornstein--Uhlenbeck semigroups},
\newblock{\it Comm. Math. Phys.} {\bf 210} {\rm (2000)}, 85-105.

\bibitem[CM]{CM} E.A. Carlen, J. Maas, Gradient flow and entropy inequalities for quantum Markov semigroups with detailed balance, J. Funct. Analysis, Vol. 273, Issue 5, 2017, Pages 1810-1869, ISSN 0022-1236, \url{https://doi.org/10.1016/j.jfa.2017.05.003}.

\bibitem[CZ]{CZ} F. Cipriani and B. Zegarlinski,  KMS Dirichlet forms, coercivity and superbounded Markovian semigroups, 
\url{doi.org/10.48550/arXiv.2105.06000}

\bibitem[FQ]{FQ} F. Fagnola, and R. Quezada, “Two photon absorption and emision process,” Infinite Dimen. Anal., Quantum Probab., Relat. Top.8, 573 (2005).

\bibitem[GM]{GM}  
D. Goderis, C. Maes,
Constructing quantum dissipations and their reversible
states from classical interacting spin systems,
Annales de l’I. H. P., section A, tome 55, no 3 (1991) 805-828
\url{http://www.numdam.org/item?id=AIHPA_1991__55_3_805_0}

%\bibitem[H]{H} Richard Herrmann, Fractional Calculus: An Introduction for Physicists,
%World Scientific, 2011.

\bibitem[LR]{LR} E. H. Lieb and D. W. Robinson,  The finite group velocity of quantum spin systems, Comm.Math. Phys. 28 (1972), 251-257.

\bibitem[MZ]{MZ} A.W. Majewski and B. Zegarlinski, 
On Quantum Stochastic Dynamics and Noncommutative Lp Spaces, Lett. Math. Phys. 36 (1996)  337-349, 
\url{https://doi.org/10.1007/BF00714401}

%\bibitem[Mat1]{Mat1} T.Matsui, The Bosonic Central Limit Theorem,  in Trends in 	Infin. Dimens. Anal. Quantum Probab. Relat. Top.  2002, 1278: 114-122 
 
\bibitem[Mat2]{Mat2} T.Matsui,
Markov Semigroups Which Describe the Time Evolution of Some Higher Spin Quantum Models,
J. Functional Analysis,
Volume 116, Issue 1, 15 August 1993, Pages 179-198
%https://doi.org/10.1006/jfan.1993.1109


%\bibitem[Mat3]{Mat3}  T.Matsui, 
%Markov Semigroups on UHF Algebras, Rev. Math. Phys.Vol. 05, No. 03, pp. 587-600 (1993)   
%https://doi.org/10.1142/S0129055X93000176

% Reviews in Mathematical PhysicsVol. 14, No. 07n08, pp. 675-700 (2002)  
%BOSONIC CENTRAL LIMIT THEOREM FOR THE ONE-DIMENSIONAL XY MODEL
%TAKU MATSUI
%https://doi.org/10.1142/S0129055X02001272

%
%Quantum Probability Communications No Access
%QUANTUM STATISTICAL MECHANICS AND FELLER SEMIGROUP
%TAKU MATSUI
%https://doi.org/10.1142/9789812816054_0004 


\bibitem[MsZ1]{MsZ1}  S. Mehta, B. Zegarlinski, Group representations in noncommutative spaces, in preparation.

\bibitem[MsZ2]{MsZ2}  S. Mehta, B. Zegarlinski,  Generalised Commutation Relations and Applications, in preparation.

\bibitem[INZ]{INZ}  J. Inglis,  M. Neklyudov, and B. Zegarlinski,   Ergodicity for Infin. Dimens. Anal. Quantum Probab. Relat. Top. Vol. 15, No. 1 (2012) 1250005, 28 pages, \url{https://doi.org/10.1142/S0219025712500051}. 

\bibitem[OZ]{OZ} R.  Olkiewicz \&  B. Zegarlinski,   Hypercontractivity in noncommutative $\mathbb{L}_p$ spaces. \newblock{\it J. Funct. Anal.} {\bf 161(1)}   {\rm (1999)} 246–285, \url{https://doi.org/10.1006/jfan.1998.3342}

\bibitem[KP]{KP} C.K. Ko, Y.M. Park, \newblock{Construction of a family of quantum Ornstein-Uhlenbeck semigroups} \newblock{\it J. Math. Phys.} {\bf 45}, {\rm (2004)}, 609--627.

%\bibitem[MOZ]{MOZ} A.W. Majewski, R. Olkiewicz, and B. Zegarlinski. "Dissipative dynamics for quantum spin systems on a lattice." J. Phys. A Math. Theor. 31.8 (1998): 2045.

\bibitem[OZa]{OZa} R. Olkiewicz,  and M. Żaba. “Dynamics of a degenerate parametric oscillator in a squeezed reservoir.” Phys. Lett. A 372 (2008): 4985-4989.

\bibitem[P1]{P1} Y.M. Park, \newblock{Construction of Dirichlet forms on standard forms of von Neumann algebras}, \newblock{\it Infin. Dimens. Anal. Quantum Probab. Relat. Top.}
{\bf 3} {\rm (2000)}, 1-14.

\bibitem[RW]{RW}
 G.A.Raggio,  and R.F.Werner, Quantum statistical mechanics of general mean field systems, Helvetica Phys. Acta 62 (1989)
\url{http://doi.org/10.5169/seals-11617} 

\bibitem[ShS]{ShS}D. Shale, W.F. Stinespring, The Quantum Harmonic Oscillator with Hyperbolic Phase Space, J. Funct. Anal. 1 (1967), 492-502.

\bibitem[SQV]{SQV}  G. Stragier, J. Quaegebeur, A. Verbeure, Quantum detailed balance, {\it Annales de l'I.H.P. Physique théorique} {\bf 41} {\rm (1984)}, 25-36.  

\bibitem[Z]{Z} B. Zegarlinski, \newblock{Analysis of Classical and Quantum Interacting Particle Systems},  pp 241 - 336 in
%\newblock{\it Quantum  interacting particle systems: Lecture notes  of the Volterra-CIRM  International  School, Trento, Italy, 23-29  September 2000 }  ”
\newblock{\it QP-PQ quantum probability and white noise analysis; v.  {\bf XIV}, eds. L. Accardi, F. Fagnola},  \newblock{World Sci. 2002 } .  

\end{thebibliography}
\end{document}